\documentclass[review]{elsarticle}
\usepackage{lineno,hyperref,mathtools}
\usepackage[letterpaper,top=2cm,bottom=2cm,left=3cm,right=3cm,marginparwidth=1.75cm]{geometry}

\usepackage{amssymb}
\usepackage{bm}
\usepackage{amssymb}
\usepackage{amsthm}
\usepackage{multirow,booktabs}
\usepackage{cases}
\usepackage{tabularx}
\usepackage{adjustbox}
\usepackage[figuresright]{rotating}
\usepackage{float}

\makeatletter

\newcommand{\Rmnum}[1]{\expandafter\@slowromancap\romannumeral #1@}
\makeatother

\newtheorem{theorem}{Theorem}

\newtheorem{remark}{Remark}
\newtheorem{definition}[theorem]{Definition}
\newtheorem{lemma}[theorem]{Lemma}
\newtheorem{corollary}[theorem]{Corollary}
\newtheorem{example}[theorem]{Example}
\newcommand{\GRS}{{\mathrm{GRS}}}
\newcommand{\Hull}{{\mathrm{Hull}}}
\newcommand{\C}{{\mathcal{C}}}
\newcommand{\F}{{\mathbb{F}}}

\begin{document}

\begin{frontmatter}

\title{MDS Codes with Euclidean and Hermitian Hulls of Flexible Dimensions and Their Applications to EAQECCs}
\tnotetext[mytitlenote]{This research was supported by the National Natural Science Foundation of China (No.U21A20428 and 12171134).
}


\author[mymainaddress]{Yang Li}
\ead{ly3650920@outlook.com}

\author[mymainaddress]{Ruhao Wan}
\ead{wanruhao98@163.com}

\author[mymainaddress]{Shixin Zhu\corref{mycorrespondingauthor}}
\cortext[mycorrespondingauthor]{Corresponding author}
\ead{zhushixinmath@hfut.edu.cn}

\address[mymainaddress]{School of Mathematics, HeFei University of Technology, Hefei 230601, China}

\begin{abstract}
  The hull of a linear code is the intersection of itself with its dual code with respect to certain inner product. 
  Both Euclidean and Hermitian hulls are of theorical and practical significance. In this paper, we construct several 
  new classes of MDS codes via (extended) generalized Reed-Solomon (GRS) codes and determine their Euclidean or Hermitian hulls. 
  Specifically, four new classes of MDS codes with Hermitian hulls of flexible dimensions and six new classes of MDS 
  codes with Euclidean hulls of flexible dimensions are constructed. For the former, we further construct four new classes of 
  entanglement-assisted quantum error-correcting codes (EAQECCs) and four new classes of MDS EAQECCs of length $n>q+1$. 
  For the latter, we also give some examples on Euclidean self-orthogonal and one-dimensional Euclidean hull MDS codes. 
\end{abstract}

\begin{keyword}
hulls\sep entanglement-assisted quantum error-correcting codes\sep generalized Reed-Solomon codes\sep extended generalized Reed-Solomon codes

\MSC[2010] 94B05\sep 81p70
\end{keyword}

\end{frontmatter}


\section{Introduction}\label{sec-introduction}
Let $\C$ be a linear code over a finite filed. Denote the dual code of $\C$ with respect to certain inner product 
by $\C^\bot$, such as the usual Euclidean inner product or classical Hermitian inner product. The hull of $\C$ is 
defined by the linear code $\C\cap \C^\bot$, denoted by $\Hull(\C)$, which was first introduced by Assmus et al. 
\cite{RefJ1} to classify finite projective planes. Over the years, numerous studies have shown that the hull of 
linear codes plays a very important role in coding theory. On one hand, the hull of linear codes is closely related 
to the complexity of algorithms for computing the automorphism group of a linear code \cite[]{RefJ2} and for checking 
permutation equivalence of two linear codes \cite[]{RefJ3,RefJ3'}. In general, these algorithms are really efficient 
when the dimension of the hull is small. Some magnificent results on linear codes with small hulls were proposed in 
\cite{RefJ15,RefJ12,RefJ13,RefJ31,RefJ30,RefJ16} by using tools such as Gaussian sums, algebraic function fields, 
partial difference sets and so on.

On the other hand, the hull of linear codes has important applications in the construction of so-called entanglement-assisted 
quantum error-correcting codes (EAQECCs). EAQECCs were introduced by Burn et al. \cite{RefJ19} and were rapidly developed 
by other scholars. By all accounts, the introduction of EAQECCs is regarded as a milestone in the development of coding 
theory. Customarily, we denote $[[n,k,d;c]]_q$ as an EAQECC, which can encode $k$ logical qubits into $n$ physical  
qubits with the help of $c$ pairs of maximally entangled states over $\F_q$ and can correct up to $\lfloor \frac{d-1}{2} \rfloor$ 
qubit-errors. Different from classical quantum error-correcting codes (QECCs), EAQECCs can be constructed by any linear code, while 
QECCs can only be constructed by special linear codes with certain self-orthogonality or satisfying certain dual containing 
condition \cite{RefJ17,RefJ18}. Similar to classical linear codes, for EAQECCs, people are also willing to construct MDS EAQECCs, 
i.e., EAQECCs that reach the quantum Singleton bound \cite{RefJ32}.

However, what needs to be emphasized is the difficulty of the computation of the number of $c$. Fortunately, in 2018, 
Guenda et al. \cite{RefJ4} described some celebrated relationships between the number of $c$ and the dimension of the 
hull of a linear code, where the hull was considered under both the Euclidean inner product and Hermitian inner product. 
After this, people worked to determine dimensions of hulls of different linear codes, and constructed new EAQECCs 
and MDS EAQECCs (e.g., see \cite{Ref xin4,RefJ7,RefJ5,RefJ4,Ref xin2,RefJ21,RefJ24,RefJ6,Ref xin3,RefJ23,RefJ20,Ref xin1,RefJ22} and references therein). 
We can summarize some outstanding works on this topic as follows.  
In \cite[]{RefJ5}, Guenda et al. completely determined all possible $q$-ary MDS EAQECCs of length $n\leq q+1$ via the 
$\ell$-intersection pair of linear codes, which implies that the construction of $q$-ary MDS EAQECCs of length $n>q+1$ will 
be the main theme of our researches. In \cite{RefJ7}, (extended) generalized Reed-Solomon (GRS) codes with Euclidean and 
Hermitian hulls of arbitrary dimensions were discussed by Fang et al.. As applications, some good MDS EAQECCs with flexible 
parameters were obtained. Recently, in \cite{RefJ27}, Wang et al. constructed some MDS EAQECCs based on GRS codes with Euclidean 
hulls of flexible dimensions and these MDS EAQECCs no longer need to use a fixed $c$. In \cite[]{RefJ34}, Chen proved that if 
an $[n,k]_{q^2}$ Hermitian self-orthogonal code exists, then $[n,k]_{q^2}$ linear codes with Hermitian hulls of arbitrary 
dimensions exist. Based on this consequence, a large number of MDS EAQECCs can be directly derived.

Inspired and motivated by these works, in this paper, we study (extended) GRS codes and determine their Euclidean or Hermitian hulls. 
Moreover, using those MDS codes with Hermitian hulls of flexible dimensions, we construct four new classes of EAQECCs and four new classes 
of MDS EAQECCs with flexible parameters. The lengths of these MDS EAQECCs are all greater than $q+1$. For reference, we list the parameters 
of some known MDS EAQECCs and the new ones in Table \ref{tab:Intro1}. Besides, as concrete examples of the Euclidean case, some Euclidean 
self-orthogonal and one-dimensional Euclidean hull MDS codes are given.

The rest of this paper is organized as follows. 
In Section \ref{sec-preliminaries}, we review some basic notations and results on (extended) GRS codes and hulls.  
In Section \ref{sec-construction}, we construct several new classes of MDS codes with Euclidean or Hermitian hulls of flexible dimensions.  
Section \ref{sec-application} constructs some new families of EAQECCs and MDS EAQECCs of length $n>q+1$. 
And finally, Section \ref{sec-conclusion} concludes this paper.

\newcommand{\tabincell}[2]{\begin{tabular}{@{}#1@{}}#2\end{tabular}}
\begin{table}
\label{tab:Intro1}
\caption{Some known constructions of MDS EAQECCs of length $n>q+1$}
\begin{center}
\resizebox{\textwidth}{80mm}{
	\begin{tabular}{ccc}
		\hline
	  	Parameters & Constraints & Ref.\\
    \hline
    $[[\frac{q^2-1}{t}, \frac{q^2-1}{t}-2d+t+2, d; t]]_q$ & $q$ odd, $t$ odd, $t\geq 3$, $t\mid q+1$, $\frac{(t-1)(q-1)}{t}+2\leq d\leq \frac{(t+1)(q-1)}{t}-2$ & \cite[]{Ref xin4}\\

    $[[q^2+1,q^2+1-q-l,q+1;q-l]]_q$ & $q=p^m\geq3$, $0\leq l\leq q$ &  \cite[]{RefJ7}\\
    $[[tr^z,tr^z-k-l,k+1;k-l]]_q$ & $q=p^m\geq 3$, $r=p^e$, $e\mid m$, $1\leq t\leq r$, $1\leq z\leq 2\frac{m}{e}-1$, $1 \leq k\leq \lfloor \frac{n-1+q}{q+1} \rfloor$, $0\leq l\leq k$ & \cite[]{RefJ7}\\
    $[[tr^z+1,tr^z+1-k-l,k+1;k-l]]_q$ & $q=p^m\geq 3$, $r=p^e$, $e\mid m$, $1\leq t\leq r$, $1\leq z\leq 2\frac{m}{e}-1$, $1 \leq k\leq \lfloor \frac{n-1+q}{q+1} \rfloor$, $0\leq l\leq k-1$ &   \cite[]{RefJ7}\\
    $[[tn',tn'-k-l,k+1;k-l]]_q$ & $q=p^m\geq 3$, $n'\mid (q^2-1)$, $1\leq t\leq \frac{q-1}{n_1}$, $n_1=\frac{n'}{\gcd(n',q+1)}$, $1 \leq k\leq \lfloor \frac{n+q}{q+1} \rfloor$, $0\leq l\leq k-1$ &   \cite[]{RefJ7}\\
    $[[tn'+1,tn'+1-k-l,k+1;k-l]]_q$ & $q=p^m\geq 3$, $n'\mid (q^2-1)$, $1\leq t\leq \frac{q-1}{n_1}$, $n_1=\frac{n'}{\gcd(n',q+1)}$, $1\leq k\leq \lfloor \frac{n+q}{q+1}\rfloor$, $0\leq l\leq k$ &   \cite[]{RefJ7}\\

    $[[\frac{q^2-1}{t}, \frac{q^2-1}{t}-4qm+4m^2+3, 2m(q-1); (2m-1)^2]]_q$ & $q\geq 3$, $t\mid q^2-1$, $1\leq m\leq \lfloor \frac{q+1}{4t}$$\rfloor$ & \cite[]{Ref xin2}\\
    $[[\frac{q^2+1}{t}, \frac{q^2+1}{t}-4qm+4q+4m^2-8m+3, 2q(m-1)+2; 4(m-1)^2+1]]_q$ & $q\geq 7$, $t\mid q^2+1$, $2\leq m\leq \lfloor \frac{q+1}{4t}\rfloor$ & \cite[]{Ref xin2}\\

    $[[lh+mr,lh+mr-2d+c,d+1; c]]_q$ & \tabincell{c}{$s\mid q+1$, $t\mid q-1$, $l=\frac{q^2-1}{s}$, $m=\frac{q^2-1}{t}$, $1\leq h\leq \frac{s}{2}$,\\ $2\leq r\leq \frac{t}{2}$, $c=h-1$, $1\leq d\leq min\{\frac{s+h}{2}\cdot \frac{q+1}{s}+2, \frac{q+1}{2}+\frac{q-1}{t}-1\}$} & \cite[]{RefJ21}\\
    $[[1+(2e+1)\frac{q^2-1}{2s+1},1+(2e+1)\frac{q^2-1}{2s+1}-2k+c,k+1;c]]_q$ & $0\leq e\leq s-1$, $(2s+1)\mid q+1$, $c=2e+1$, $1\leq k\leq (s+1+e)\frac{q+1}{2s+1}-1$ & \cite[]{RefJ21}\\
    $[[1+(2e+1)\frac{q^2-1}{2s},1+(2e+1)\frac{q^2-1}{2s}-2k+c,k+1;c]]_q$ & $0\leq e\leq s-2$, $2s\mid q+1$, $c=2e+2$, $1\leq k\leq (s+1+e)\frac{q+1}{2s}-1$ & \cite[]{RefJ21}\\
    $[[1+(2e+1)\frac{q^2-1}{2s},1+(2e+1)\frac{q^2-1}{2s}-2k+c,k+1;c]]_q$ & $0\leq e\leq s-1$, $2s\mid q+1$, $c=2e+1$, $1\leq k\leq (s+e)\frac{q+1}{2s}-2$ & \cite[]{RefJ21}\\

    $[[n,n-k-h,k+1;k-h]]_q$ & $q>3$, $m>1$, $m\mid q$, $1<k\leq \lfloor \frac{n}{2}\rfloor$, $n+k>m+1$, $1\leq n\leq m$, $1\leq h\leq n-m+k-1$& \cite[]{RefJ24}\\
    $[[n,n-k-h,k+1;k-h]]_q$ & \tabincell{c}{$q>3$, $m>1$, $m\mid q$, $1<k\leq \lfloor \frac{n}{2}\rfloor$,\\ $2n-k-1<m<2n-1$, $1\leq n\leq m$, $1\leq h\leq 2n-m-1$} & \cite[]{RefJ24}\\

    $[[n, n-k-l,k+1; k-l]]_q$ & $q+1<n<2(q-1)$, $n-q<k<\lfloor \frac{n}{2}\rfloor$, $1\leq l\leq k+q-n$ & \cite[]{Ref xin3}\\

    $[[\frac{q+1}{7}(q-1),\frac{q+1}{7}(q-1)+5-2d,d;3]]_q$ & $d\leq \frac{n+2}{2}$, $\frac{5(q+1)}{7}\leq d \leq \frac{6(q+1)}{7}-2$ & \cite[]{RefJ23}\\
    $[[\frac{q+1}{7}(q-1),\frac{q+1}{7}(q-1)+7-2d,d;5]]_q$ & $d\leq \frac{n+2}{2}$, $\frac{6(q+1)}{7}\leq d \leq q$ & \cite[]{RefJ23}\\
    $[[\frac{q+1}{7}(q-1),\frac{q+1}{7}(q-1)+9-2d,d;7]]_q$ & $q$ odd, $d\leq \frac{n+2}{2}$, $d=\frac{8(q+1)}{7}-1$ & \cite[]{RefJ23}\\
    $[[\frac{q+1}{4}(q-1),\frac{q+1}{4}(q-1)+4-2d,d;2]]_q$ & $q$ odd, $d\leq \frac{n+2}{2}$, $\frac{3(q+1)}{4}\leq d \leq q$ & \cite[]{RefJ23}\\
    $[[\frac{q+1}{4}(q-1),\frac{q+1}{4}(q-1)+6-2d,d;4]]_q$ & $q$ odd, $d\leq \frac{n+2}{2}$, $q+1\leq d \leq \frac{5(q+1)}{4}-1$ & \cite[]{RefJ23}\\
    $[[\frac{q+1}{6}(q-1),\frac{q+1}{6}(q-1)+4-2d,d;2]]_q$ & $q$ odd, $d\leq \frac{n+2}{2}$, $\frac{4(q+1)}{6}\leq d \leq \frac{5(q+1)}{6}-1$ & \cite[]{RefJ23}\\
    $[[\frac{q+1}{6}(q-1),\frac{q+1}{6}(q-1)+6-2d,d;4]]_q$ & $q$ odd, $d\leq \frac{n+2}{2}$, $\frac{5q+1}{6}\leq d \leq q$ & \cite[]{RefJ23}\\

    $[[q^2+1,q^2-2\delta,2\delta+2; 2\delta+1]]_q$ & $q$ odd, $s=\frac{n}{2}$, $r\mid q-1$, $r\nmid q+1$, $0\leq \delta\leq \frac{(r-1)(s-1)}{r}$ & \cite[]{RefJ20}\\
    $[[q^2+1,q^2-2\delta-1,2\delta+3; 2\delta+2]]_q$ & $q=2^m\ (m\geq 1)$, $\mu=\frac{n-r}{2}$, $r\mid q-1$, $r\nmid q+1$, $0\leq \delta\leq \frac{\mu-1}{r}$ & \cite[]{RefJ20}\\
    $[[q^2+1,q^2-2\delta,2\delta+2; 2\delta+1]]_q$ & $q=2^m\ (m\geq 1)$, $r\mid q-1$, $r\nmid q+1$, $0\leq \delta\leq \frac{(r-1)(n-2)}{2r}$ & \cite[]{RefJ20}\\
    
    $[[q^2+1, q^2-4(m-1)(q-m-1), 2(m-1)q+2; 4(m-1)^2+1]]_q$ & $q\geq 5$, $2\leq m\leq \frac{q-1}{2}$ & \cite[]{Ref xin1}\\

    $[[n,n-k-l,k+1;k-l]]_q$ & $q=p^m\geq 3$, $1\leq k\leq q-1$, $q^2-k\leq n\leq q^2$, $0\leq l\leq n+k-q^2$ & new\\
    $[[n+1,n+1-k-l,k+1;k-l]]_q$ & $q=p^m\geq 3$, $1\leq k\leq q-1$, $q^2-k+1\leq n\leq q^2$, $0\leq l\leq n+k-q^2-1$ & new\\
    $[[n+1,n-1-q-l,q+1;q-l]]_q$ & $q=p^m\geq 3$, $q^2-q\leq n\leq q^2$, $0\leq l\leq n+q-q^2$ & new\\
    $[[m(q-1)+1,m(q-1)+1-k-l,k+1;k-l]]_q$ & $q=p^m\geq 3$, $2\leq m\leq q$, $1\leq k\leq m-1$, $0\leq l\leq k-1$ & new\\
		\hline
	\end{tabular}}
\end{center}
\end{table}

\section{Preliminaries}\label{sec-preliminaries}

Let $q$ be a prime power and $\F_q$ be the finite field with $q$ elements. For any positive integer $n$, $\F_q^n$ can be seen as 
an $n$-dimensional vector space over $\F_q$. Then a $k$-dimensional subspace of $\F_q^n$ with minimum distance $d$ is just a linear 
code $\C$, denoted by $[n,k,d]_q$. A linear code $\C$ is called an MDS code if $d=n-k+1$. Now, we review some basic notations and 
results on (extended) GRS codes and hulls.  

Firstly, for any two vectors $\mathbf{x}=(x_1,x_2,\dots,x_n)$ and $\mathbf{y}=(y_1,y_2,\dots,y_n)$ of $\F_q^n$, we can define different 
inner products between them. Specifically, the Euclidean inner product between $\mathbf{x}$ and $\mathbf{y}$ is defined by
\[\begin{split}
\langle \mathbf{x},\mathbf{y}\rangle _E=\sum_{i=1}^{n}x_iy_i.
\end{split} \]
If $\C$ is a linear code of length $n$ over $\F_q$, then the Euclidean dual code of $\C$, denoted by $\C^{\bot_E}$, 
can be described as the set  
\[\begin{split}
  \C^{\bot_E}=\{\mathbf{x}\in \F_{q}^n:\ \langle \mathbf{x},\mathbf{y}\rangle _E=0,~{\rm for~all}~\mathbf{y}\in \C\}.
\end{split}\]

The Hermitian inner product between $\mathbf{x}$ and $\mathbf{y}$ is defined by 
\[\begin{split}
    \langle \mathbf{x},\mathbf{y}\rangle_H=\sum_{i=1}^nx_iy_i^q.
\end{split}\]
If $\C$ is a linear code of length $n$ over $\F_{q^2}$, then the Hermitian dual code of $\C$, denoted by $\C^{\bot_H}$, 
can be similarly described as the set  
\[\begin{split}
\C^{\bot_H}=\{\mathbf{x}\in \F_{q^2}^n:\ \langle \mathbf{x},\mathbf{y}\rangle _H=0,~{\rm for~all}~\mathbf{y}\in \C\}.
\end{split}\]

Then as said before, we define the Euclidean hull (resp. Hermitian hull) of $\C$ as $\C\cap \C^{\bot_E}$ (resp. $\C\cap \C^{\bot_H}$), 
denoted by $\Hull_E(\C)$ (resp. $\Hull_H(\C)$). It is well known that $\C$ is an Euclidean (resp. Hermitian) self-orthogonal code if 
$\Hull_E(\C)=\C$ (resp. $\Hull_H(\C)=\C$). More generally, for a positive integer $l$, if $\dim(\Hull_E(\C))=l$ (resp. $\dim(\Hull_H(\C))=l$), 
we call $\C$ a $l$-dimensional Euclidean (resp. Hermitian) hull code. 

Denote $\F_q^*=\F_q\backslash \{0\}$. Choose $n$ distinct elements $a_1,a_2,\cdots,a_n$ from $\F_q$ and $n$ nonzero elements $v_1,v_2,\cdots,v_n$ from 
$\F^*_q$. As special MDS codes, GRS codes and extended GRS codes can be defined as follows. 
Set $\mathbf{a}=(a_1,a_2,\dots,a_n)$ and $\mathbf{v}=(v_1,v_2,\dots,v_n)$. The GRS code of length $n$ associated to $\mathbf{a}$ and $\mathbf{v}$, denoted by 
$\GRS(\mathbf{a},\mathbf{v})$, is defined by 
\[\begin{split}
  \GRS(\mathbf{a},\mathbf{v})=\{(v_1f(a_1),v_2f(a_2),\dots,v_nf(a_n)):\ f(x)\in \F_q[x]\ {\rm{and}} \deg(f(x))\leq k-1\},
\end{split} \]
where the elements $a_1,a_2,\dots,a_n$ are called the code locators of $\GRS(\mathbf{a},\mathbf{v})$ and $v_1,v_2,\dots,v_n$ are called the column 
multipliers of $\GRS(\mathbf{a},\mathbf{v})$.

With a practical technique, the extended GRS code of length $n+1$ associated to $\mathbf{a}$ and $\mathbf{v}$, denoted by $\GRS_k(\mathbf{a},\mathbf{v},\infty)$, can 
be derived. Specifically, the definition of $\GRS_k(\mathbf{a},\mathbf{v},\infty)$ is 
\[\begin{split}
  \GRS_k(\mathbf{a},\mathbf{v},\infty)=\{(v_1f(a_1),v_2f(a_2),\dots,v_nf(a_n),f_{k-1}):\ f(x)\in \F_q[x]\ {\rm{and}} \deg(f(x))\leq k-1\},
\end{split} \]
where $f_{k-1}$ is the coefficient of $x^{k-1}$ in $f(x)$. 

For our purposes, considering both (extended) GRS codes and hulls, some basic results need to be introduced. To this end, for $0\leq i\leq n$, we denote 
\begin{align}\label{eq.Intro.ui}
  u_i=\prod_{1\leq j\leq n,j\neq i}(a_i-a_j)^{-1},
\end{align}
which will appear frequently in this paper and is critical to our constructions. Then the coming results can help us calculate the dimension of 
the hull of a GRS code or an extended GRS code. 

\begin{lemma}\label{lem.Euclidean hulls}(\cite{RefJ28}) 
  Considering the Euclidean inner product over $\F_q$, the following statements hold. 
  \begin{enumerate} 
    \item [{\rm (1)}] A codeword $\boldsymbol{c}=(v_1f(a_1),v_2f(a_2),\dots,v_nf(a_n))$ of $\GRS_k(\mathbf{a},\mathbf{v})$ 
    is contained in $\GRS_k(\mathbf{a},\mathbf{v})^{\bot_E}$ if and only if there exists a polynomial $g(x)\in \F_q[x]$ 
    with $\deg(g(x))\leq n-k-1$ such that 
    $$(v_1^2f(a_1),v_2^2f(a_2),\dots,v_n^2f(a_n))=(u_1g(a_1),u_2g(a_2),\dots,u_ng(a_n)).$$

    \item [{\rm (2)}] A codeword $\boldsymbol{c}=(v_1f(a_1),v_2f(a_2),\dots,v_nf(a_n),f_{k-1})$ of $\GRS_k(\mathbf{a},\mathbf{v},\infty)$ 
    is contained in $\GRS_k(\mathbf{a},\mathbf{v},\\\infty)^{\bot_E}$ if and only if there exists a polynomial $g(x)\in \F_q[x]$ with 
    $\deg(g(x))\leq n-k$ such that 
    $$(v_1^2f(a_1),v_2^2f(a_2),\dots,v_n^2f(a_n),f_{k-1})=(u_1g(a_1),u_2g(a_2),\dots,u_ng(a_n),-g_{n-k}),$$
    where $g_{n-k}$ is the coefficient of $x^{n-k}$ in $g(x)$.
  \end{enumerate}
\end{lemma}

\begin{lemma}\label{lem.Hermitian hulls}(\cite{RefJ29})   
  Considering the Hermitian inner product over $\F_{q^2}$, the following statements hold. 
  \begin{enumerate} 
    \item [{\rm (1)}] A codeword $\boldsymbol{c}=(v_1f(a_1),v_2f(a_2),\dots,v_nf(a_n))$ of $\GRS_k(\mathbf{a},\mathbf{v})$ 
    is contained in $\GRS_k(\mathbf{a},\mathbf{v})^{\bot_H}$ if and only if there exists a polynomial $g(x)\in \F_{q^2}[x]$ 
    with $\deg(g(x))\leq n-k-1$ such that 
    $$(v_1^{q+1}f^q(a_1),v_2^{q+1}f^q(a_2),\dots,v_n^{q+1}f^q(a_n))=(u_1g(a_1),u_2g(a_2),\dots,u_ng(a_n)).$$
    
    \item [{\rm (2)}] A codeword $\boldsymbol{c}=(v_1f(a_1),v_2f(a_2),\dots,v_nf(a_n),f_{k-1})$ of $\GRS_k(\mathbf{a},\mathbf{v},\infty)$ 
    is contained in $\GRS_k(\mathbf{a},\mathbf{v},\\\infty)^{\bot_H}$ if and only if there exists a polynomial $g(x)\in \F_{q^2}[x]$ 
    with $\deg(g(x))\leq n-k$ such that 
    $$(v_1^{q+1}f^q(a_1),v_2^{q+1}f^q(a_2),\dots,v_n^{q+1}f^q(a_n),f_{k-1}^q)=(u_1g(a_1),u_2g(a_2),\dots,u_ng(a_n),-g_{n-k}),$$
    where $g_{n-k}$ is the coefficient of $x^{n-k}$ in $g(x)$.
  \end{enumerate}
\end{lemma}

For a given $l$-dimensional Hermitian hull linear code, by Corollary $2.2$ of \cite[]{RefJ34}, linear codes with Hermitian hulls of 
flexible dimensions can be obtained in an explicit way. For convenience, we equivalently write it in the following form.
\begin{lemma}\label{lem.Hermitian hulls from l-dim Hermitian hull codes}(\cite[]{RefJ34}) 
  Let $\C$ be an $[n, k]_{q^2}$ linear code with $l$-dimensional Hermitian hull. Then there exists an $[n,k]_{q^2}$ linear code with 
  $l'$-dimensional Hermitian hull for nonnegative integer $l'$ satisfying $0\leq l' \leq l$.
\end{lemma}

In particular, according to the proof of Lemma \ref{lem.Hermitian hulls from l-dim Hermitian hull codes} in \cite[]{RefJ34}, 
for the MDS case and $l=k-1$, we can precisely derive the following corollary.

\begin{corollary}\label{coro.Hermitian hulls from Hermitian almost self-orthogonal MDS codes}
  Let $\C$ be an $[n, k]_{q^2}$ MDS code with $(k-1)$-dimensional Hermitian hull. Then there exists an $[n,k]_{q^2}$ MDS code with 
  $l'$-dimensional Hermitian hull for nonnegative integer $l'$ satisfying $0\leq l' \leq k-1$.
\end{corollary}

Finally, we make some conventions. As readers may have noticed, throughout this paper, 
\begin{itemize}
  \item when we talk about Euclidean inner product or Euclidean hull, the finite field that matches is always $\F_q$; 
  \item when we talk about Hermitian inner product or Hermitian hull, the finite field that matches is always $\F_{q^2}$. 
\end{itemize}
For the multiplication sign "$\prod_{i=a}^{b} \cdot$" over $\F_q$, we make the following agreement: 
\begin{itemize}
  \item if $a\leq b$, then the operation is performed according to the standard multiplication over $\F_q$; 
  \item if $a>b$, then the result of this operation is always $1$, where $1$ is the unit element of $\F_q$. 
\end{itemize}
In addition, we need to emphasize the following notations: 
\begin{itemize}
  \item $\lfloor x \rfloor$ denotes the largest integer less than or equal to $x$;
  \item $\lceil x \rceil$ denotes the smallest integer greater than or equal to $x$;
  \item $[a,b]$ denotes a set of $x$ satisfying $a\leq x\leq b$, where $a\leq b$.
\end{itemize}

\section{Constructions} \label{sec-construction}
In this section, we present several new classes of MDS codes via (extended) GRS codes, whose Euclidean hulls or Hermitian hulls are entirely determined. 
We also introduce a new method to construct $[n,k]_{q^2}$ MDS codes with $(k-1)$-dimensional Hermitian hull. For Euclidean cases, some Euclidean self-orthogonal 
and one-dimensional Euclidean hull MDS codes are given as examples. 

\subsection{Construction A for MDS codes with Hermitian hulls of flexible dimensions}\label{constructionA}

Denote $\F_{q^2}=\{a_1,a_2,\cdots,a_n, a_{n+1},\cdots,a_{q^2}\}.$ It is clear that $a_i\neq a_j$ for any $1\leq i\neq j\leq q^2$. Note that 
\begin{align*}
    \prod_{1\leq j\leq q^2,j\neq i}(a_i-a_j)^{-1}=\prod_{x\in \F^*_{q^2}}x=-1, 
\end{align*}
then $u_i$ defined as Eq. (\ref{eq.Intro.ui}) can be further denoted by  
\begin{equation}\label{eq.ui over Fq2}
    u_i=\prod_{1\leq j\leq n,j\neq i}(a_i-a_j)^{-1}=-\prod_{j=n+1}^{q^2}(a_i-a_j).
\end{equation}

Based on the basic fact above, three new classes of MDS codes with Hermitian hulls of flexible dimensions can be constructed as follows.  
\begin{theorem}\label{th.ConA.1}
Let $q=p^m\geq 3$, for $1\leq k\leq q-1$, if $q^2-k\leq n\leq q^2$, then there exists an $[n,k]_{q^2}$ MDS code 
with $l$-dimensional Hermitian hull, where $0\leq l\leq n+k-q^2$.\\
\end{theorem}

\begin{proof}
Let notations be the same as before. Denote $s=n+k-q^2-l$. Take $\mathbf{a}=(a_1,a_2,\dots,a_n)$ and $\mathbf{v}=(v_1,v_2,\dots,v_s,1,\dots ,1)$, 
where $v_i^{q+1}\neq 1$ for all $1\leq i\leq s$. Note that $(q^2-1)\nmid (q+1)$ for any $q$ satisfying 
$q\geq 3$, then such $\mathbf{v}$ does exist. We now consider the Hermitian hull of the $q^2$-ary $[n,k]$ MDS code $\C=\GRS_k(\mathbf{a},\mathbf{v})$. 

For any codeword $$\boldsymbol{c}=(v_1f(a_1),\dots,v_sf(a_s),f(a_{s+1}),\dots,f(a_n))\in \Hull_H(\C),$$ by the result (1) of Lemma \ref{lem.Hermitian hulls}, 
there exists a polynomial $g(x)\in \F_{q^2}[x]$ with $\deg(g(x))\leq n-k-1$ such that
\begin{align}\label{eq.ConA.1}
  \begin{split}
    & (v_1^{q+1}f^q(a_1),\dots,v_s^{q+1}f^q(a_s),f^q(a_{s+1}),\dots,f^q(a_n)) \\
    = & (u_1g(a_1),\dots,u_sg(a_s),u_{s+1}g(a_{s+1}),\dots,u_ng(a_n)). \\    
  \end{split}
\end{align}

On one hand, from the last $n-s$ coordinates of Eq. (\ref{eq.ConA.1}) and Eq. (\ref{eq.ui over Fq2}), we have
\[\begin{split}
f^q(a_i)=u_ig(a_i)=-\prod_{j=n+1}^{q^2}(a_i-a_j)g(a_i),\ s+1 \leq i\leq n. 
\end{split} \]
It follows that $f^q(x)=-\prod_{j=n+1}^{q^2}(x-a_j)g(x)$ has at least $n-s$ distinct roots. Recall that $s=n+k-q^2-l$ and $1\leq k\leq q-1$, then 
\[\begin{split}
& \deg(f^q(x))\leq q(k-1)\leq q^2-k-1=n-s-l-1\leq n-s-1,\\
& \deg(\prod_{j=n+1}^{q^2}(x-a_j)g(x))\leq (q^2-n)+(n-k-1)=q^2-k-1\leq n-s-1.
\end{split} \]
Hence, we can conclude that $f^q(x)=-\prod_{j=n+1}^{q^2}(x-a_j)g(x)$ from the fact $n-s-1<n-s$. Moreover, $\prod_{j=n+1}^{q^2}(x-a_j)\mid f^q(x)$. 

On the other hand, from the first $s$ coordinates of Eq. (\ref{eq.ConA.1}), we have
\[\begin{split}
v_i^{q+1}f^q(a_i)=u_ig(a_i)=f^q(a_i),\ 1\leq i\leq s.
\end{split} \]
For any $1\leq i\leq s$, since $v_i^{q+1}\neq 1$, we have $f(a_i)=0$. Therefore, $f(x)$ can be written as 
\begin{align*} 
    f(x)=h(x)\prod_{j=n+1}^{q^2}(x-a_j)\prod_{i=1}^{s}(x-a_i),
\end{align*}
where $h(x)\in \F_{q^2}[x]$ with $\deg(h(x))\leq n+k-q^2-s-1$. It deduces that $\dim(\Hull_H(\C))\leq n+k-q^2-s$.

Conversely, let $f(x)$ be a polynomial of form $h(x)\prod_{j=n+1}^{q^2}(x-a_j)\prod_{i=1}^{s}(x-a_i)$, where $h(x)\in \F_{q^2}[x]$ 
with $\deg(h(x))\leq n+k-q^2-s-1$. Take $g(x)=-\prod_{j=n+1}^{q^2}(x-a_j)^{-1}f^q(x)$, then $g(x)$ is a polynomial in $\F_{q^2}[x]$ 
with $\deg(g(x))\leq q(k-1)-(q^2-n)\leq n-k-1$. Moreover, by Eq. (\ref{eq.ui over Fq2}),  
we have 
\[\begin{split}
 & (v_1^{q+1}f^q(a_1),\dots,v_s^{q+1}f^q(a_s),f^q(a_{s+1}),\dots,f^q(a_n))\\ 
=& (u_1g(a_1),\dots,u_sg(a_s),u_{s+1}g(a_{s+1}),\dots,u_ng(a_n)). \\
\end{split} \]
According to the result (1) of Lemma \ref{lem.Hermitian hulls}, the vector $$(v_1f(a_1),\dots,v_sf(a_s),f(a_{s+1}),\dots,f(a_n)) \in \Hull_H(\C).$$ 
It deduces that $\dim(\Hull_H(\C))\geq n+k-q^2-s$.

In summary, $\dim(\Hull_H(\C))=n+k-q^2-s=l$. This completes the proof.
\end{proof}

\begin{theorem}\label{th.ConA.2}
Let $q=p^m\geq 3$, for $1\leq k\leq q-1$, if $q^2-k+1\leq n\leq q^2$, then there exists an $[n+1,k]_{q^2}$ MDS code with 
$l$-dimensional Hermitian hull, where $0\leq l\leq n+k-q^2-1$.
\end{theorem}
\begin{proof}

  Denote $s=n+k-q^2-l-1$ and let other notations be the same as Theorem \ref{th.ConA.1}. We now consider the Hermitian hull of the $q^2$-ary $[n+1,k]$ 
  MDS code $\C=\GRS_k(\mathbf{a},\mathbf{v},\infty)$. 
  
  For any codeword $$\boldsymbol{c}=(v_1f(a_1),\dots,v_sf(a_s),f(a_{s+1}),\dots,f(a_n),f_{k-1})\in \Hull_H(\C),$$ 
  by the result (2) of Lemma \ref{lem.Hermitian hulls}, there exists a polynomial $g(x)\in \F_{q^2}[x]$ with $\deg(g(x))\leq n-k$ such that
  \begin{align}\label{eq.ConA2.1}
    \begin{split}
      & (v_1^{q+1}f^q(a_1),\dots,v_s^{q+1}f^q(a_s),f^q(a_{s+1}),\dots,f^q(a_n),f^q_{k-1}) \\
      = & (u_1g(a_1),\dots,u_sg(a_s),u_{s+1}g(a_{s+1}),\dots,u_ng(a_n),-g_{n-k}).
    \end{split}
  \end{align}
  
  On one hand, from the last $n-s+1$ coordinates of Eq. (\ref{eq.ConA2.1}) and Eq. (\ref{eq.ui over Fq2}), we have
  \[\begin{split}
  f^q(a_i)=u_ig(a_i)=-\prod_{j=n+1}^{q^2}(a_i-a_j)g(a_i),\ s+1 \leq i\leq n\ {\rm{and}}\ f^q_{k-1}=-g_{n-k}. 
  \end{split} \]
  It follows that $f^q(x)=-\prod_{j=n+1}^{q^2}(x-a_j)g(x)$ has at least $n-s$ distinct roots. Recall that $s=n+k-q^2-l-1$ and $1\leq k\leq q-1$, then 
  \[\begin{split}
  & \deg(f^q(x))\leq q(k-1)\leq q^2-k=n-s-l-1\leq n-s-1,\\
  & \deg(\prod_{j=n+1}^{q^2}(x-a_j)g(x))\leq (q^2-n)+(n-k)=q^2-k\leq n-s-1.
  \end{split} \]
  Hence, we can conclude that $f^q(x)=-\prod_{j=n+1}^{q^2}(x-a_j)g(x)$ from the fact $n-s-1<n-s$. Moreover, $\prod_{j=n+1}^{q^2}(x-a_j)\mid f^q(x)$. Now, we determine the value of $f_{k-1}$. 
  If $f_{k-1}\neq 0$, since $f^q(x)=-\prod_{j=n+1}^{q^2}(x-a_j)g(x)$ and $f^q_{k-1}=-g_{n-k}$, we have $q(k-1)=(q^2-n)+(n-k)$, which contradicts to $1\leq k\leq q-1$. Hence, $f_{k-1}=0$ and $\deg(f(x))\leq k-2$. 
  
  On the other hand, similar to the proof of Theorem \ref{th.ConA.1}, it follows that $f(a_i)=0$ for any $1\leq i\leq s$ from the first $s$ coordinates of Eq. (\ref{eq.ConA2.1}). 
  Therefore, $f(x)$ can be written as 
  \begin{align*} 
      f(x)=h(x)\prod_{j=n+1}^{q^2}(x-a_j)\prod_{i=1}^{s}(x-a_i),
  \end{align*}
  where $h(x)\in \F_{q^2}[x]$ with $\deg(h(x))\leq n+k-q^2-s-2$. It deduces that $\dim(\Hull_H(\C))\leq n+k-q^2-s-1$.
  
  Conversely, let $f(x)$ be a polynomial of form $h(x)\prod_{j=n+1}^{q^2}(x-a_j)\prod_{i=1}^{s}(x-a_i)$, where $h(x)\in \F_{q^2}[x]$ 
  with $\deg(h(x))\leq n+k-q^2-s-2$. Take $g(x)=-\prod_{j=n+1}^{q^2}(x-a_j)^{-1}f^q(x)\in \F_{q^2}[x]$, then $g(x)$ is a polynomial in 
  $\F_{q^2}[x]$ with $\deg(g(x))\leq q(k-2)-(q^2-n)\leq n-k-1$. Moreover, by Eq. (\ref{eq.ui over Fq2}), we have 
  \[\begin{split}
    & (v_1^{q+1}f^q(a_1),\dots,v_s^{q+1}f^q(a_s),f^q(a_{s+1}),\dots,f^q(a_n),0)\\
  = & (u_1g(a_1),\dots,u_sg(a_s),u_{s+1}g(a_{s+1}),\dots,u_ng(a_n),0).\\
  \end{split} \]
  According to the result (2) of Lemma \ref{lem.Hermitian hulls}, the vector $$(v_1f(a_1),\dots,v_sf(a_s),f(a_{s+1}),\dots,f(a_n),0) \in \Hull_H(\C).$$ 
  It deduces that $\dim(\Hull_H(\C))\geq n+k-q^2-s-1$.
  
  In summary, $\dim(\Hull_H(\C))=n+k-q^2-s-1=l$. This completes the proof.
  \end{proof}

  In particular, considering $k=q$, one more new classes of MDS codes with Hermitian hulls of more flexible dimensions can be obtained from the following way. 
  \begin{theorem}\label{th.ConA.3}
    Let $q=p^m\geq 3$, for any $q^2-q\leq n\leq q^2$, there exists an $[n+1,q]_{q^2}$ MDS code with $l$-dimensional Hermitian hull, where $0\leq l\leq n+q-q^2$.
    \end{theorem}
    \begin{proof}
    Denote $s=n+q-q^2-l$ and let other notations be the same as before. We now consider the Hermitian hull of the $q^2$-ary $[n+1,q]$ MDS code 
    $\C=\GRS_k(\mathbf{a},\mathbf{v},\infty)$. 
      
      For any codeword $$\boldsymbol{c}=(v_1f(a_1),\dots,v_sf(a_s),f(a_{s+1}),\dots,f(a_n),f_{q-1})\in \Hull_H(\C),$$ 
      by the result (2) of Lemma \ref{lem.Hermitian hulls}, there exists a polynomial $g(x)\in \F_{q^2}[x]$ with $\deg(g(x))\leq n-q$ such that
      \begin{align}\label{eq.ConA3.1}
        \begin{split}
          & (v_1^{q+1}f^q(a_1),\dots,v_s^{q+1}f^q(a_s),f^q(a_{s+1}),\dots,f^q(a_n),f^q_{q-1}) \\
          = & (u_1g(a_1),\dots,u_sg(a_s),u_{s+1}g(a_{s+1}),\dots,u_ng(a_n),-g_{n-q}).
        \end{split}
      \end{align}
      
      On one hand, from the last $n-s+1$ coordinates of Eq. (\ref{eq.ConA3.1}) and Eq. (\ref{eq.ui over Fq2}), we have
      \[\begin{split}
      f^q(a_i)=u_ig(a_i)=-\prod_{j=n+1}^{q^2}(a_i-a_j)g(a_i),\ s+1 \leq i\leq n,\ {\rm{and}}\ f^q_{q-1}=-g_{n-q}.  
      \end{split} \]
      Note that 
      \[\begin{split}
      & \deg(f^q(x))\leq q(q-1)=q^2-q,\\
      & \deg(\prod_{j=n+1}^{q^2}(x-a_j)g(x))\leq (q^2-n)+(n-q)=q^2-q. 
      \end{split} \]
      Since $f^q_{q-1}=-g_{n-q}$, then $\deg(f^q(x))=q^2-q$, i.e., $f_{q-1}\neq 0$ is equivalent to $\deg(\prod_{j=n+1}^{q^2}(x-a_j)g(x))=q^2-q$, i.e., $g_{n-q}\neq 0$. 
      It follows that $$\deg(f^q(x)+\prod_{j=n+1}^{q^2}(x-a_j)g(x))\leq q^2-q-1\leq n-s-l-1\leq n-s-1$$ for any possible $f(x)$. Hence, we can conclude that 
      $f^q(x)=-\prod_{j=n+1}^{q^2}(x-a_j)g(x)$ from the fact $n-s-1<n-s$. Moreover, $\prod_{j=n+1}^{q^2}(x-a_j)\mid f^q(x)$.
      
      On the other hand, taking a similar manner to Theorem \ref{th.ConA.2}, $f(x)$ can be written as 
      \begin{align*} 
          f(x)=h(x)\prod_{j=n+1}^{q^2}(x-a_j)\prod_{i=1}^{s}(x-a_i),
      \end{align*}
      where $h(x)\in \F_{q^2}[x]$ with $\deg(h(x))\leq n+q-q^2-s-1$. It deduces that $\dim(\Hull_H(\C))\leq n+q-q^2-s$.
      
      Conversely, let $f(x)$ be a polynomial of form $h(x)\prod_{j=n+1}^{q^2}(x-a_j)\prod_{i=1}^{s}(x-a_i)$, where $h(x)\in \F_{q^2}[x]$ with $\deg(h(x))\leq n+q-q^2-s-1$.
      Take $g(x)=-\prod_{i=n+1}^{q^2}(x-a_i)^{-1}f^q(x)\in \F_{q^2}[x]$, then $g(x)$ is a polynomial in $\F_{q^2}[x]$ with $\deg(g(x))\leq q(q-1)-(q^2-n)=n-q$. Moreover, 
      by Eq. (\ref{eq.ui over Fq2}), we have   
      \[\begin{split}
        & (v_1^{q+1}f^q(a_1),\dots,v_s^{q+1}f^q(a_s),f^q(a_{s+1}),\dots,f^q(a_n),f^q_{q-1})\\
      = & (u_1g(a_1),\dots,u_sg(a_s),u_{s+1}g(a_{s+1}),\dots,u_ng(a_n),-g_{n-q}).
      \end{split} \]
      According to the result (2) of Lemma \ref{lem.Hermitian hulls}, the vector $$(v_1f(a_1),\dots,v_sf(a_s),f(a_{s+1}),\dots,f(a_n),f_{q-1}) \in \Hull_H(\C).$$ 
      It deduces that $\dim(\Hull_H(\C))\geq n+q-q^2-s$.
      
      In summary, $\dim(\Hull_H(\C))=n+q-q^2-s=l$. This completes the proof.
      \end{proof}

\begin{remark}\label{rem1.conA.comparion}
  \begin{enumerate}
    \item [{\rm (1)}] Discuss the range of lengths in Theorems \ref{th.ConA.1}, \ref{th.ConA.2} and \ref{th.ConA.3} as follows:
    \begin{itemize}
      \item For Theorem \ref{th.ConA.1}: It follows from $1\leq k\leq q-1$ that $n\geq q^2-k\geq q^2-q+1$. Since $q\geq 3$, we have $n>q+1$;
      \item For Theorem \ref{th.ConA.2}: It follows from $1\leq k\leq q-1$ that $n\geq q^2-k+1\geq q^2-q+2$. Since $q\geq 3$, we have $n>q+1$;
      \item For Theorem \ref{th.ConA.3}: It follows from $q\geq 3$ that $n\geq q^2-q> q+1$. 
    \end{itemize}
    Hence, lengths in Theorems \ref{th.ConA.1}, \ref{th.ConA.2} and \ref{th.ConA.3} are at least $q^2-q+1$ and must be greater than $q+1$.
    \item [{\rm (2)}] Note that the following facts: 
    \begin{itemize}
      \item the dimension $l$ of the Hermitian hull can take $n+k-q^2$ in Theorem \ref{th.ConA.1};
      \item the length $n+1$ of the MDS codes can take $q^2+1$ in Theorem \ref{th.ConA.2};
      \item the dimension $k$ of the MDS codes can and only can take $q$ in Theorem \ref{th.ConA.3}. 
    \end{itemize}
    Hence, MDS codes constructed by Theorems \ref{th.ConA.1}, \ref{th.ConA.2} and \ref{th.ConA.3} will not be exactly the same as each other. 
  \end{enumerate}
\end{remark}

\subsection{Construction B for MDS codes with Hermitian hulls of arbitrary dimensions}\label{constructionB}
In this subsection, we construct a new classes of MDS codes with Hermitian hulls of arbitrary dimensions. 
To this end, we need the following Lemma. 

\begin{lemma}\label{lem.Hermitian almost self-orthogonal EGRS}
  Let $\GRS_k(\mathbf{a},\mathbf{v},\infty)$ be an extended GRS code associated to 
  $\mathbf{a}$ and $\mathbf{v}$ over $\F_{q^2}$. Let $a_i$ and $v_i$ be the $i$-th 
  elements of $\mathbf{a}$ and $\mathbf{v}$, respectively, where $1\leq i\leq n$.
  Assume that there exists a monic polynomial $h(x)\in \F_{q^2}[x]$ with $\deg(h(x))\leq q+n-(q+1)k$ such that 
\begin{align}\label{eq.Hermitian almost self-orthogonal EGRS}
  \lambda u_ih(a_i)=v_i^{q+1},\ 1\leq i\leq n,
\end{align}
where $\lambda\in \F_{q^2}^*$. If $\deg(h(x))=q+n-(q+1)k$ and $\lambda=-1$ do not hold at the same time, 
then $\GRS_k(\mathbf{a},\mathbf{v},\infty)$ is a $(k-1)$-dimensional Hermitian hull MDS code of length $n$. 
\end{lemma}

  \begin{proof}
  Multiplying both sides of Eq. (\ref{eq.Hermitian almost self-orthogonal EGRS}) by $f^q(a_i)$, we have 
    \begin{equation}\label{eq.Hermitian almost self-orthogonal EGRS.2}
      \lambda u_ih(a_i)f^{q}(a_i)=v_i^{q+1}f^{q}(a_i),\ 1\leq i\leq n,
    \end{equation}
    where $\lambda\in \F_{q^2}^*$.
  
  Let $g(x)=\lambda h(x)f^{q}(x)$ and substitute $g(a_i)$ into Eq. (\ref{eq.Hermitian almost self-orthogonal EGRS.2}), then 
  \begin{align*}
      \deg(g(x))\leq q+n-(q+1)k+q(k-1)=n-k
  \end{align*}
  and 
  \begin{equation*}
    u_ig(a_i)=v_i^{q+1}f^{q}(a_i),\ 1\leq i\leq n.
  \end{equation*}

  We discuss the relationship between $f_{k-1}$ and $g_{n-k}$ in two ways. 
  
  $\bullet$ $\textbf{Case 1:}$ When $\deg(f(x))<k-1$, we have $\deg(f(x))\leq k-2$ and $f_{k-1}=0$. Note that, in this case,    
  \begin{equation*}
    \deg(g(x))\leq q+n-(q+1)k+q(k-2)=n-k-q<n-k,
  \end{equation*}
  which implies $g_{n-k}=0$. Hence, $f^q_{k-1}=0=-g_{n-k}$ and the fact has nothing to do with $\lambda$.
  
  $\bullet$ $\textbf{Case 2:}$ When $\deg(f(x))=k-1$, we have $f_{k-1}\neq 0$ and 
  \begin{equation*}
    \deg(g(x))\leq q+n-(q+1)k+q(k-1)=n-k.
  \end{equation*}
    Note that $\deg(g(x))=n-k$ if and only if $\deg(h(x))=q+n-(q+1)k$. 
    \begin{itemize}
      \item $\textbf{Subcase 2.1:}$ When $\deg(h(x))<q+n-(q+1)k$, we still have $\deg(g(x))<n-k$ and $g_{n-k}=0$. 
      Then $f^q_{k-1}\neq -g_{n-k}$ for the fact $f_{k-1}\neq 0$. 

      \item $\textbf{Subcase 2.2:}$ When $\deg(h(x))=q+n-(q+1)k$, i.e., $\deg(g(x))=n-k$ and $g_{n-k}\neq 0$. Note that $h(x)$ is 
      a monic polynomial over $\F_{q^2}$, then the coefficients of the highest degree of $g(x)$ and $\lambda h(x)f^q(x)$ are $g_{n-k}$ 
      and $\lambda f^{q}_{k-1}$, respectively. Since $g(x)=\lambda h(x)f^{q}(x)$, we have $f_{k-1}^{q}=\lambda g_{n-k}$. Therefore, it 
      is easy to see that $f_{k-1}^{q}=-g_{n-k}$ if and only if $\lambda=-1$.
    \end{itemize}
  
  In summary, when $\deg(h(x))=q+n-(q+1)k$ and $\lambda=-1$ do not hold at the same time, $f^q_{k-1}=-g_{n-k}$ follows if and only if $\deg(f(x))<k-1$, 
  if and only if $f_{k-1}=-g_{n-k}=0$. Hence, by the result (2) of Lemma \ref{lem.Hermitian hulls}, for any polynomial $f(x)\in \F_{q^2}[x]$, we have  
  \begin{align*}
   \mathbf{c}=(v_1f(a_1),v_2f(a_2),\dots,v_nf(a_n),0)\in \Hull_H(\GRS_k(\mathbf{a},\mathbf{v},\infty)),\ {\rm{where}}\ \deg(f(x))\leq k-2 
  \end{align*}
  and 
  \begin{align*}
  \mathbf{c}=(v_1f(a_1),v_2f(a_2),\dots,v_nf(a_n),f_{k-1})\notin \Hull_H(\GRS_k(\mathbf{a},\mathbf{v},\infty)),\ {\rm{where}}\ \deg(f(x))= k-1.
  \end{align*}
  Clearly, we can conclude that $\dim(\Hull_H(\GRS_k(\mathbf{a},\mathbf{v},\infty)))=k-1$. The desired result follows.
\end{proof}    

Moreover, a criterion for an extended GRS code being Hermitian self-orthogonal can be derived. We write it in the following corollary and one can finish the proof 
in a similar manner to Lemma \ref{lem.Hermitian almost self-orthogonal EGRS}.

\begin{corollary}\label{coro.Hermitian self-orthogonal EGRS}
  Let $\GRS_k(\mathbf{a},\mathbf{v},\infty)$ be an extended GRS code associated to $\mathbf{a}$ and $\mathbf{v}$ over $\F_{q^2}$. Let $a_i$ and $v_i$ be the $i$-th 
  elements of $\mathbf{a}$ and $\mathbf{v}$, respectively, where $1\leq i\leq n$. If there exists a monic polynomial $h(x)\in \F_{q^2}[x]$ with $\deg(h(x))=q+n-(q+1)k$ 
  such that
  \begin{equation}\label{eq.Hermitian self-orthogonal EGRS}
    -u_ih(a_i)=v_i^{q+1},\ 1\leq i\leq n,
  \end{equation}
  then $\GRS_k(\mathbf{a},\mathbf{v},\infty)$ is a Hermitian self-orthogonal MDS code of length $n$.
\end{corollary}

As we all know, if $\Hull_H(\C)=\C$, i.e., $\dim(\C)=\dim(\Hull_H(\C))=k$, then $\C$ is called a Hermitian self-orthogonal code. 
For a unified description, we introduce the definition of Hermitian almost self-orthogonal codes as follows. 

\begin{definition}\label{def.Hermitian almost self-orthogonal}
For a linear code $\C$ of dimension $k$, if $\dim(\Hull_H(\C))=k-1$, then we call $\C$ a Hermitian 
almost self-orthogonal code. In particular, if $\C$ is Hermitian almost self-orthogonal and MDS, then 
we call $\C$ a Hermitian almost self-orthogonal MDS code. 
\end{definition}

Now, we consider a kind of decompositions of $\F_{q^2}$, which was first introduced in \cite[]{RefJ35}. 
Assume $n=m(q-1)$, $1\leq m\leq q$. Set 
\begin{equation}
  \begin{split}
    I_t & = \{x\cdot \omega^{t-1} :\ x\in \F_q^*\} \\
        & = \{\omega^{t-1},\omega^{q+1+t-1},\cdots,\omega^{(q+1)i+t-1},\cdots,\omega^{(q+1)(q-2)+t-1}\}, \\
  \end{split}
\end{equation}
where $1\leq t\leq q$. Denote $$\mathcal{I} =\bigcup_{t=1}^{m}I_t=\{a_1,a_2,\cdots,a_n\},$$ where $a_i=\omega^{(q+1)i+t-1}$ for $1\leq i\leq n$. 
From \cite[]{RefJ35}, $u_i$ can be written as certain concrete form for different $q$. We rephrase the result in Lemma \ref{ConB.lemma_ui}.

\begin{lemma}\label{ConB.lemma_ui}(\cite[]{RefJ35})
Let notations be the same as before.
\begin{enumerate}
  \item [{\rm (1)}] If $q$ is even, then 
  \begin{equation}\label{ConB.eq1}
    \begin{split}
      u_i = & (q-1)^{-1}a_i^{2-q}\omega^{(t-1)(m-2)+\frac{m(m+1)}{2}-(q+1)i_0} \\
          = & (q-1)^{-1}\omega^{(t-1)(m-q)+\frac{m(m+1)}{2}-(q+1)[(q-2)i+i_0]} \\
    \end{split}
  \end{equation}
  for some integer $i_0$. 

  \item [{\rm (2)}] If $q$ is odd, then 
  \begin{equation}\label{ConB.eq2}
    \begin{split}
      u_i = & (q-1)^{-1}a_i^{2-q}\omega^{(t-1)(m-2)+\frac{(m-q-1)(m-1)}{2}-(q+1)i_0} \\
          = & (q-1)^{-1}\omega^{(t-1)(m-q)+\frac{(m-q-1)(m-1)}{2}-(q+1)[(q-2)i+i_0]} \\
    \end{split}
  \end{equation}
  for some integer $i_0$.
\end{enumerate}
\end{lemma}

\begin{theorem}\label{th.Hermitian almost self-orthogonal.ConA.}
Let $q$ be a prime power and $2\leq m\leq q$. Then there exists an $[m(q-1)+1,k]_{q^2}$ 
Hermitian almost self-orthogonal MDS code, where $1\leq k\leq m-1$. 
\end{theorem}
\begin{proof}
Let notations be the same as before. We continue our proof in the following two cases.

$\textbf{Case 1:}$ When $q$ is even, set $\lambda = (q-1)\omega^{-\frac{m(m+1)}{2}}$. By Eq. (\ref{ConB.eq1}), $u_i$ can be written as 
\[\begin{split}
 u_i = & (q-1)^{-1}\omega^{(t-1)(m-q)+\frac{m(m+1)}{2}-(q+1)[(q-2)i+i_0]}\\
     = & \lambda^{-1} \omega^{(t-1)(m-q)-(q+1)[(q-2)i+i_0]}
\end{split} \]
for some integer $i_0$. Let $h(x)=x^{q-m}$, then according to the form of $a_i$, we have   
\[\begin{split}
    \lambda u_ih(a_i) = & \omega^{(t-1)(m-q)-(q+1)[(q-2)i+i_0]}\cdot \omega^{(t-1)(q-m)+(q+1)(q-m)i}\\
    = & \omega^{-(q+1)[(m-2)i+i_0]}.\\
\end{split} \]
It is obviously that  $\lambda u_ih(a_i)\in \F_q^*$ and there exists $v_i\in \F_{q^2}^*$ such that $v_i^{q+1}=\lambda u_ih(a_i)$ for each $1\leq i\leq n$. 
Set $\mathbf{v}=(v_1,v_2,\dots,v_n)$.

Note that $\lfloor \frac{n+m}{q+1} \rfloor=\lfloor \frac{m(q+1)-m}{q+1}\rfloor=m-1\geq 1$ for $2\leq m\leq q$. Therefore, for $1\leq k\leq \lfloor \frac{n+m}{q+1} \rfloor=m-1$, 
we can consider the extended GRS code $\GRS_k(\mathbf{a},\mathbf{v},\infty)$ associated to $\mathbf{a}$ and $\mathbf{v}$, where $\mathbf{a}$ and $\mathbf{v}$ are defined 
as above. Furthermore, it is easy to see that 
\[\begin{split}
  \deg(h(x))=q-m\leq q+n-(q+1)k 
\end{split} \]
and the equation holds if and only if $k=m-\frac{m}{q+1}$, which contradicts to the fact that $k$ is an integer. It follows that 
$\deg(h(x))< q+n-(q+1)k$, thus, according to Lemma \ref{lem.Hermitian almost self-orthogonal EGRS} and Definition \ref{def.Hermitian almost self-orthogonal}, 
$\GRS_k(\mathbf{a},\mathbf{v},\infty)$ is an $[m(q-1)+1,k]_{q^2}$ Hermitian almost self-orthogonal MDS code.

$\textbf{Case 2:}$ When $q$ is odd, set $\lambda = (q-1)\omega^{\frac{(q-m+1)(m-1)}{2}}$. 
Taking a similar manner to $\textbf{Case 1}$ above, let $h(x)=x^{q-m}$ again, then we can 
also obtain an $[m(q-1)+1,k]_{q^2}$ Hermitian almost self-orthogonal MDS code, where $1\leq k\leq m-1$. 

Combining $\textbf{Case 1}$ and $\textbf{Case 2}$, the desired result follows.
\end{proof}

\begin{remark}\label{rem1.lambda isnot equal to -1}
  One of the most critical steps in the proof of Theorem \ref{th.Hermitian almost self-orthogonal.ConA.} 
  is to show that $\deg(h(x))<q+n-(q+1)k$. In fact, by Lemma \ref{lem.Hermitian almost self-orthogonal EGRS}, we can also prove $\lambda\neq -1$. 
  And a simple calculation shows that $q\geq 3$ is required if we finish the proof of Theorem \ref{th.Hermitian almost self-orthogonal.ConA.}
  in this way. Therefore, this specifical example shows us that the new method introduced in Lemma \ref{lem.Hermitian almost self-orthogonal EGRS} 
  is flexible for using and that it may produce a small difference when we use it from different views.  
\end{remark}

Applying Corollary \ref{coro.Hermitian hulls from Hermitian almost self-orthogonal MDS codes}, we can draw the following result.

\begin{corollary}\label{coro.1.flexible Hermitian hulls from 1.Hermitian almost self-orthogonal}
  Let $q$ be a prime power and $2\leq m\leq q$. Then for $1\leq k\leq m-1$, there exists an $[m(q-1)+1,k]_{q^2}$ 
  MDS code with $l$-dimensional Hermitian hull, where $0\leq l\leq k-1$.  
\end{corollary}

Note that all MDS codes in  Corollary \ref{coro.1.flexible Hermitian hulls from 1.Hermitian almost self-orthogonal} have type $[m(q-1)+1,k]_{q^2}$. 
As far as we know, for a similar form of length, there are some constructions of Hermitian self-orthogonal MDS codes in previous studies. 
We list them in Table \ref{tab:know Hermitian self-orthogonal MDS codes}.
According to Lemma \ref{lem.Hermitian hulls from l-dim Hermitian hull codes}, MDS codes with $l$-dimensional Hermitian hull for each 
$0\leq l\leq k$ can be derived from these Hermitian self-orthogonal MDS codes. Therefore, it is necessary to compare these results with our construction. 
In Remark \ref{rem1.comparing Hermitian self-orthogonal}, we discuss their differences in detail.

\begin{table}[!htb]
  \caption{Some known Hermitian self-orthogonal MDS codes}
  \label{tab:know Hermitian self-orthogonal MDS codes}       
  \begin{small}
  \begin{center}
    \begin{tabular}{ccccc}
      \hline
       Class & $q$ & $n$ & $k$ & Ref.\\ 
      \hline 

      1 & $q\geq 3$ & \tabincell{c}{$n=r\frac{q^2-1}{2s+1}+1$,\\ $(2s+1)\mid (q+1)$, $1\leq r\leq 2s+1$} & $1\leq k\leq (s+1)\frac{q+1}{2s+1}-1$ & \cite[]{RefJB1}\\
       
      2 & $q\geq 3$ & \tabincell{c}{$n=r\frac{q^2-1}{2s}$+1,\\ $2s\mid (q+1)$, $2\leq r\leq 2s$} & $1\leq k\leq (s+1)\frac{q+1}{2s}-1$ & \cite[]{RefJB1} \\

      3 & $q\geq 3$ & \tabincell{c}{$n=tn'+1$, $1\leq t\leq \frac{q-1}{n_1}$,\\ $n_1=\frac{n'}{\gcd(n',q+1)}$, $n'\mid (q^2-1)$} & $1\leq k\leq \lfloor \frac{n+q}{q+1} \rfloor$ & \cite[]{RefJ7} \\

      4 & \tabincell{c}{$0\leq r\leq q$,\\ $q+1\equiv r({\rm{mod}}\ 2r)$} & $n=r(q-1)+1$ & $k\leq \frac{q-1+r}{2}$ & \cite[]{RefJB3} \\

      5 & \tabincell{c}{$t\geq 1$, $q\equiv -1({\rm{mod}}\ 2t+1)$} & \tabincell{c}{$n=1+\frac{r(q^2-1)}{2t+1}$,\\ $0\leq r\leq 2t+1$, $\gcd(r,q)=1$} & $k\leq \frac{t+1}{2t+1}\cdot q-\frac{t}{2t+1}$ & \cite[]{RefJB4} \\

      6 & Arbitrary & \tabincell{c}{$n=\frac{q^2-1}{t}+1$,\\ $t=2w+1$, $t\mid (q+1)$} & $2\leq k\leq \frac{t+1}{2t}(q-1)$  & \cite[]{RefJB5} \\

      \hline 
    \end{tabular}
  \end{center}
\end{small}
  \end{table}

\begin{remark}\label{rem1.comparing Hermitian self-orthogonal}
  From the perspective of length and dimension, we can easily illustrate that the MDS codes with Hermitian hulls of arbitrary dimensions constructed by 
  Theorem \ref{th.Hermitian almost self-orthogonal.ConA.} are new in general. Recall that our MDS codes have type $[m(q-1)+1,k]_{q^2}$, where $2\leq m\leq q$ and $1\leq k\leq m-1$. 
  \begin{enumerate}
    \item [{\rm (1)}] \textbf{Comparisons of lengths:}
    \begin{itemize}
      \item For Class $3$:\\
      Taking $n'=q-1$, if $q$ is even, then $\gcd(q-1,q+1)=1$ and $n_1=q-1$; if $q$ is odd, then $\gcd(q-1,q+1)=2$ and 
      $n_1=\frac{q-1}{2}$. Hence, Class $3$ can generate 
      $[t(q-1)+1,k]_{q^2}$ Hermitian self-orthogonal MDS codes, where $q$ is even and $t=1$; 
      $[t(q-1)+1,k]_{q^2}$ Hermitian self-orthogonal MDS codes, where $q$ is odd and $1\leq t\leq 2$. 
      However, $2\leq m\leq q$, thus, most lengths in our construction are new. 
      \item For Classes $4$ and $5$:\\
      Note that more conditions are required for $q$ or $r$, which confirms our assertion. 
      \item For Class $6$:\\
      Rewrite $n=\frac{q+1}{t} \cdot (q-1)+1$, where $t=2w+1$ is odd and $t\mid (q+1)$. Clearly, the value set of 
      $\frac{q+1}{t}$ is narrow relative to the value set of $m$ in our construction. In fact, all even $t$ and 
      $t\nmid (q+1)$ are not suitable to Class $6$.
    \end{itemize} 

      \item [{\rm (2)}] \textbf{Comparisons of dimensions:}
      \begin{itemize}
      \item For Classes $1$ and $2$:\\
      Take $m=\frac{r(q+1)}{2s+1}$ and $\frac{r(q+1)}{2s}$ in Class $1$ and Class $2$, respectively. Then the corresponding 
      ranges of the dimension in our construction should be $1\leq k\leq \frac{r(q+1)}{2s+1}-1$ and $\frac{r(q+1)}{2s}-1$. 
      Considering the ranges of $r$, almost half of our MDS codes can take larger dimensions in general. 
    \end{itemize}
    Of course, the restrictions of $2(s+1)\mid (q+1)$ and $2s\mid (q+1)$ may also limit the lengths of Classes $1$ and $2$ in some cases. 
    For example, MDS codes of lengths $25,\ 41,\ 57$ and $73$ over $\F_{3^4}$ can not be constructed from Class $1$ and all possible MDS codes 
    over $\F_q$ with even $q$ can not be constructed from Class $2$.
  \end{enumerate}
\end{remark}

\begin{remark}\label{rem.conB.length>q+1}
  The length of MDS codes in Corollary \ref{coro.Hermitian hulls from Hermitian almost self-orthogonal MDS codes} has form $m(q-1)+1$. 
  On one hand, since $2\leq m\leq q$, then the length of MDS codes are always $m(q-1)+1\geq 2(q-1)+1\geq q+1$. And only when $q=2$, 
  the equation holds.
  On the other hand, $m(q-1)+1\leq q(q-1)+1=q^2-q+1$. According to Remark \ref{rem1.conA.comparion}, the length of MDS codes from 
  Theorems \ref{th.ConA.1}, \ref{th.ConA.2} and \ref{th.ConA.3} is at least $q^2-q+1$. Hence, almost all MDS codes from Corollary 
  \ref{coro.Hermitian hulls from Hermitian almost self-orthogonal MDS codes} are new with respect to MDS codes from  
  Theorems \ref{th.ConA.1}, \ref{th.ConA.2} and \ref{th.ConA.3}. 
\end{remark}

\begin{example}\label{exam1.Hermitian almost self-orthogonal MDS codes}
  For some different $q$, we give some concrete examples of Hermitian almost self-orthogonal MDS codes from 
  Theorem \ref{th.Hermitian almost self-orthogonal.ConA.}. For clarity, we list them in Table \ref{tab:2}. 
\end{example}

\begin{table}[!htb]
\caption{Some examples of new Hermitian almost self-orthogonal MDS codes from Theorem \ref{th.Hermitian almost self-orthogonal.ConA.}}
\label{tab:2}       
\begin{center}
	\begin{tabular}{cccc|cccc}
		\hline
	   $q$ & $m$ & $n$ & $k$ & $q$ & $m$ & $n$ & $k$\\
		\hline
    4 & 3 & 10 & [1,2] & 4 & 4 & 13 & [1,3]\\

    8 & 5 & 36 & [1,4] & 8 & 6 & 43 & [1,5]\\
    8 & 7 & 50 & [1,6] & 8 & 8 & 57 & [1,7]\\

    25 & 17 & 409 & [1,16] & 25 & 18 & 433 & [1,17]\\
    25 & 21 & 505 & [1,20] & 25 & 22 & 529 & [1,21]\\

    27 & 11 & 287 & [1,10] & 27 & 12 & 313 & [1,11]\\
    27 & 23 & 599 & [1,22] & 27 & 24 & 625 & [1,23]\\
		\hline
	\end{tabular}
\end{center}
\end{table}

\subsection{Construction C for MDS codes with Euclidean hulls of flexible dimensions}\label{constructionC}
In this construction, we consider the Euclidean inner product over $\F_{q}$. Note that $(q-1)\nmid 2$ for any $q>3$, then there exists $v_i\in \F_{q}^*$ 
such that $v_i^2\neq 1$. By similar ways to Theorems \ref{th.ConA.1}, \ref{th.ConA.2} and \ref{th.ConA.3}, the dimension of the Euclidean hull of 
corresponding MDS codes can be determined. We present these results in Theorem \ref{th.ConA_form of Euclidean} and omit the proof. 

\begin{theorem}\label{th.ConA_form of Euclidean}
  Let $q=p^m>3$. The following statements hold.
  \begin{enumerate}
    \item [(1)] For $1\leq k\leq \lfloor \frac{q}{2} \rfloor$, if $q-k\leq n\leq q$, then there exists an $[n,k]_q$ MDS code with $l$-dimensional Euclidean hull, where $0\leq l\leq n+k-q$. 
    \item [(2)] For $1\leq k\leq \lfloor \frac{q}{2} \rfloor$, if $q-k+1\leq n\leq q$, then there exists an $[n+1,k]_q$ MDS code with $l$-dimensional Euclidean hull, where $0\leq l\leq n+k-q-1$.
    \item [(3)] For odd $q$ and any $\frac{q-1}{2}\leq n\leq q$, then there exists an $[n+1,\frac{q+1}{2}]_q$ MDS code with $l$-dimensional Euclidean hull, where $0\leq l\leq n-\frac{q-1}{2}$.  
  \end{enumerate}
\end{theorem}

Besides, another two classes of MDS codes with Euclidean hulls of flexible dimensions can be exactly construted. To this end, we need to introduce a fact. 
Denote $\F_q=\{a_1,a_2,\cdots,a_n,a_{n+1},\cdots,a_q\}$. Clearly, these $q$ elements are distinct. With similar reasoning as Eq. (\ref{eq.ui over Fq2}), 
we can further express $u_i$ as 
\begin{align}\label{eq.ui over Fq}
  u_i=\prod_{1\leq j\leq n, j\neq i}=(a_i-a_j)^{-1}=-\prod_{j=n+1}^{q}(a_i-a_j)
\end{align}
over $\F_q$. 

\begin{theorem}\label{th.ConC}
  Let $q=p^m$ be a prime power. The following statements hold.
  \begin{enumerate}
    \item [{\rm (1)}]  For $1\leq k\leq \lfloor \frac{q}{2} \rfloor$, if $\lceil \frac{q}{2} \rceil \leq n\leq \min\{q-k, \lceil \frac{q}{2} \rceil+k\}$, 
    then there exists an $[n,k]$ MDS code with $l$-dimensional Euclidean hull, where $0\leq l\leq n- \lceil \frac{q}{2} \rceil$.

    \item [{\rm (2)}] For $1\leq k\leq \lfloor \frac{q+1}{2} \rfloor$, if $\lceil \frac{q+1}{2} \rceil \leq n\leq \min\{q-k+1, \lceil \frac{q+1}{2} \rceil+k-1\}$, 
    then there exists an $[n+1,k]$ MDS code with $l$-dimensional Euclidean hull, where $0\leq l\leq n- \lceil \frac{q+1}{2} \rceil$ and $l\neq n-\frac{q+1}{2}$.

    \item [{\rm (3)}] If $q$ is odd, for $1\leq k\leq \frac{q+1}{2}$, if $\frac{q+1}{2} \leq n\leq \min\{q-k+1,  \frac{q+1}{2}+k-1\}$, 
    then there exists an $[n+1,k]$ MDS code with $(n-\frac{q+1}{2}+1)$-dimensional Euclidean hull.
  \end{enumerate}
  \end{theorem}
  \begin{proof}
 (1) Let notations be the same as before. Denote $s=q-n-k+l$. 
 Take $\mathbf{a}=(a_1,a_2,\dots,a_{n})$ and $\mathbf{v}=(v_1,v_2,\dots,v_s,v_{s+1},\dots,v_n)$, where $v_i=\prod_{j=n+1}^{n+s}(a_i-a_{j})$ for any $1\leq i\leq n$.
  We now consider the Euclidean hull of the $q$-ary $[n,k]$ MDS code $\C=\GRS_k(\mathbf{a},\mathbf{v})$.

  For any codeword $$\boldsymbol{c}=(v_1f(a_1),v_2f(a_2),\dots,v_nf(a_n))\in \Hull_E(\C),$$
  by the result $(1)$ of Lemma \ref{lem.Euclidean hulls}, there exists a polynomial $g(x)\in \F_q[x]$ with $\deg(g(x))\leq n-k-1$ such that
  \[\begin{split}
  (v_1^2f(a_1),v_2^2f(a_2),\dots,v_n^2f(a_n))=(u_1g(a_1),u_2g(a_2),\dots,u_ng(a_n)),
  \end{split} \]
  which implies that $v_i^2f(a_i)=u_ig(a_i)$ for $1\leq i\leq n$. Since
  \[\begin{split}
  v_i^2f(a_i)=\prod_{j=n+1}^{n+s}(a_i-a_{j})^2f(a_i)
  \end{split} \]
  and by Eq. (\ref{eq.ui over Fq}),
  \[\begin{split}
  u_ig(a_i)=-\prod_{j=n+1}^{q}(a_i-a_j)g(a_i) ,
  \end{split} \]
  it follows that
  \[\begin{split}
  \prod_{j=n+1}^{n+s}(a_i-a_{j})f(a_i)=-\prod_{j=n+s+1}^{q}(a_i-a_j)g(a_i),
  \end{split} \]
  for $1\leq i\leq n$. Note that
  \[\begin{split}
  & \deg(\prod_{j=n+1}^{n+s}(x-a_{j})f(x))\leq s+k-1=q-n+l-1\leq n-1, \\
  & \deg(-\prod_{j=n+s+1}^{q}(x-a_j)g(x))\leq (q-n-s)+(n-k-1)=q-s-k-1=n-l-1\leq n-1.
  \end{split} \]
  Hence, we can derive that
  $$\prod_{j=n+1}^{n+s}(x-a_{j})f(x)=-\prod_{j=n+s+1}^{q}(x-a_j)g(x)$$
  from the fact $n-1<n$. Moreover, since $(\prod_{j=n+1}^{n+s}(x-a_{j}), \prod_{j=n+s+1}^{q}(x-a_{j}))=1$, we have
  $$\prod_{j=n+s+1}^{q}(x-a_j)\mid f(x).$$
  Therefore, $f(x)$ can be written as
  $$f(x)=h(x)\prod_{j=n+s+1}^{q}(x-a_{j}),$$
  where $\deg(h(x))\leq k-1-(q-n-s)=n-q+k+s-1$. It deduces that $\dim(\Hull_E(\C))\leq n-q+k+s$.

  Conversely, let $f(x)$ be a polynomial of form $h(x)\prod_{j=n+s+1}^{q}(x-a_{j})$, where $h(x)\in \F_q[x]$ and $\deg(h(x))\leq n-q+k+s-1$.
  Take
  $$g(x)=-\prod_{j=n+s+1}^{q}(x-a_{j})^{-1} \prod_{j=n+1}^{n+s}(x-a_{j})f(x)=-h(x)\prod_{j=n+1}^{n+s}(x-a_{j}),$$
  then $g(x)$ is a polynomial in $\F_q[x]$ with $\deg(g(x))\leq (n-q+k+s-1)+s\leq n-k-1$. Moreover, by Eq. (\ref{eq.ui over Fq}), we have 
  \[\begin{split}
  (v_1^2f(a_1),v_2^2f(a_2),\dots,v^2_nf(a_n))=(u_1g(a_1),u_2g(a_2),\dots,u_ng(a_n)).
  \end{split} \]
  According to the result (1) of Lemma \ref{lem.Euclidean hulls}, the vector
  $$(v_1f(a_1),v_2f(a_2),\dots,v_nf(a_n))\in \Hull_E(\C).$$
  It deduces that $\dim(\Hull_E(\C))\geq n-q+k+s$.

 In summary, we have $\dim(\Hull_E(\C))=n-q+k+s=l$. This completes the proof.

  (2) Denote $s=q-n-k+l+1$ and let other notations be the same as before. 
   We now consider the Euclidean hull of the $q$-ary $[n+1,k]$ MDS code $\C=\GRS_k(\mathbf{a},\mathbf{v},\infty)$.
  
   For any codeword $$\boldsymbol{c}=(v_1f(a_1),v_2f(a_2),\dots,v_nf(a_n),f_{k-1})\in \Hull_E(\C),$$ by the result (2) of Lemma \ref{lem.Euclidean hulls}, 
   there exists a polynomial $g(x)\in \F_q[x]$ with $\deg(g(x))\leq n-k$ such that
  \[\begin{split}
  (v_1^2f(a_1),v_1^2f(a_2),\dots,v_n^2f(a_n), f_{k-1})=(u_1g(a_1),u_2g(a_2),\dots,u_ng(a_n),-g_{n-k}), 
  \end{split} \]
Similar to the proof of (1) above, we have 
\[\begin{split}
  \prod_{j=n+1}^{n+s}(a_i-a_{j})f(a_i)=-\prod_{j=n+s+1}^{q}(a_i-a_j)g(a_i)
\end{split} \]
for $1\leq i\leq n$. Note that 
  \[\begin{split}
  & \deg(\prod_{j=n+1}^{n+s}(x-a_{j})f(x))\leq s+k-1=q-n+l\leq n-1, \\
  & \deg(-\prod_{j=n+s+1}^{q}(x-a_j)g(x))\leq (q-n-s)+(n-k)=q-s-k=n-l-1\leq n-1.
  \end{split} \]
  Hence, we can derive that
  $$\prod_{j=n+1}^{n+s}(x-a_{j})f(x)=-\prod_{j=n+s+1}^{q}(x-a_j)g(x)$$
  and $\prod_{j=n+s+1}^{q}(x-a_j)\mid f(x)$ for the same reasoning as (1) above. Now, we determine the value of $f_{k-1}$. 
  If $f_{k-1}\neq 0$, then $s+k-1=(q-n-s)+(n-k)$, which contradicts to $l\neq n-\frac{q+1}{2}$.  
  Thus $f_{k-1}=0$ and $\deg(f(x))\leq k-2$.    
  Therefore, $f(x)$ can be written as
  $$f(x)=h(x)\prod_{j=n+s+1}^{q}(x-a_{j}),$$
  where $\deg(h(x))\leq k-2-(q-n-s)=n-q+k+s-2$. It deduces that $\dim(\Hull_E(\C))\leq n-q+k+s-1$.
  
 Conversely, let $f(x)$ be a polynomial of form $h(x)\prod_{j=n+s+1}^{q}(x-a_{j})$, where $h(x)\in \F_q[x]$ and $\deg(h(x))\leq n-q+k+s-2$.
  Take
  $$g(x)=-\prod_{j=n+s+1}^{q}(x-a_{j})^{-1} \prod_{j=n+1}^{n+s}(x-a_{j})f(x)=-h(x)\prod_{j=n+1}^{n+s}(x-a_{j}),$$
  then $g(x)$ is a polynomial in $\F_q[x]$ with $\deg(g(x))\leq (n-q+k+s-2)+s\leq n-k-1 $. Moreover, by  Eq. (\ref{eq.ui over Fq}), we have 
  \[\begin{split}
  (v_1^2f(a_1),v_2^2f(a_2),\dots,v^2_nf(a_n),0)=(u_1g(a_1),u_2g(a_2),\dots,u_ng(a_n),0).
  \end{split} \]
  According  to the result (2) of Lemma \ref{lem.Euclidean hulls}, the vector
  $$(v_1f(a_1),v_2f(a_2),\dots,v_nf(a_n),0)\in \Hull_E(\C).$$
  It deduces that $\dim(\Hull_E(\C))\geq n-q+k+s-1$.
  
  In summary, we have $\dim(\Hull_E(\C))=n-q+k+s-1=l$. This completes the proof.

  (3) Let notations be the same as before. 
  We now consider the $q$-ary $[l+\frac{q+1}{2}+1,k]$ MDS code $\C=\GRS_k(\mathbf{a},\mathbf{v},\infty)$, where $q$ is odd. From the proof of (2) above, we can easily conclude that 
  the dimension of the Euclidean hull of $\C$ is $l+1$. In other words, there exists an $[n+1,k]_{q}$ MDS code with $(n-\frac{q+1}{2}+1)$-dimensional Euclidean hull. This completes the proof.
\end{proof}

  Taking $l=k$ and $l=1$ in Theorems \ref{th.ConA_form of Euclidean} and \ref{th.ConC}, we can obtain some Euclidean self-orthogonal and one-dimensional Euclidean hull MDS codes. 
  In particular, for MDS codes $\C$ with dimension $k=1$, it is easy to see that $\C$ is Euclidean self-orthogonal if and only if $\dim(\Hull_E(\C))=1$. Hence, we only consider the 
  MDS codes with dimension $k\geq 2$ in the following examples.

  \begin{example}\label{exam.Euclidean self-orthogonal}
    Taking $l=k$ in Theorem \ref{th.ConA_form of Euclidean}, we have the following Euclidean self-orthogonal MDS codes, which can also be obtained from \cite[]{RefJ6}.
    \begin{itemize}
      \item [{\rm (1)}] $[q,k]_q$ MDS codes, where $1\leq k \leq \lfloor \frac{q}{2} \rfloor$ and $q>3$;
      \item [{\rm (2)}] $[q+1,\frac{q+1}{2}]_q$ MDS codes, where $q>3$ is odd.
    \end{itemize}
    Taking $l=k$ in Theorem \ref{th.ConC}, we have the following Euclidean self-orthogonal MDS codes. As far as we know, they are new. 
    \begin{itemize}
      \item [{\rm (3)}] $[\lceil \frac{q}{2} \rceil+k,k]_q$ MDS codes, where $1\leq k \leq \lfloor \frac{q}{4} \rfloor$ and $q\geq 4$;
      \item [{\rm (4)}] $[k+\frac{q+1}{2},k]_q$ MDS codes, where $1\leq k\leq \lfloor \frac{q+3}{4} \rfloor$ and $q\geq 3$ is odd.
    \end{itemize}

  \end{example}
  
  \begin{example}\label{exam.Euclidean one-dimensional hull}
  Let $q=p^m$ be a prime power, then the dimension of Euclidean hulls of the following MDS codes is at most $1$. 
  \begin{itemize}
   \item[{\rm (1)}] $[q+1-k,k]_q$ MDS codes, where $2\leq k\leq \lfloor \frac{q}{2} \rfloor$ and $q\geq 4$;
   \item[{\rm (2)}] $[q+3-k,k]_q$ MDS codes, where $2\leq k\leq \lfloor \frac{q}{2} \rfloor$ and $q\geq 4$;
   \item[{\rm (3)}] $[\lceil \frac{q}{2} \rceil+1,k]_q$ MDS codes, where $2\leq k\leq \lfloor \frac{q}{2} \rfloor-1$ and $q\geq 6$;
   \item[{\rm (4)}] $[\lceil \frac{q+1}{2} \rceil+2,k]_q$ MDS codes, where $2\leq k\leq \lfloor \frac{q+1}{2} \rfloor-1$ and $q\geq 6$ is even;
   \item[{\rm (5)}] $[\frac{q+3}{2},k]_q$ MDS codes, where $2\leq k\leq  \frac{q+1}{2}$ and $q\geq 3$ is odd;
  \end{itemize}
  Moreover, since $\dim(\Hull_E(\C))=\dim(\Hull_E(\C^{\bot_E}))$, we can deduce that all Euclidean dual codes of MDS codes listed in $(1)-(5)$ are also 
  one-dimensional Euclidean hull MDS codes.
  \end{example}

\section{Application to EAQECCs} \label{sec-application}
\subsection{Quantum codes}
In this subsection, we introduce some basic notions about quantum codes. Let $\mathbb{C}$ be the complex field and $\mathbb{C}^q$ be the $q$-dimensional 
Hilbert space over $\mathbb{C}$. A qubit is actually a non-zero vector of $\mathbb{C}^q$. Denote a basis of $\mathbb{C}^q$ by $\{|a\rangle:\ a\in \F_q \}$, 
then a qubit $|v\rangle$ can be written as 
$$|v\rangle =\sum_{a\in \F_q}v_a|a\rangle,$$ where $v_a\in \mathbb{C}$.
 
Let $(\mathbb{C})^{\otimes n}(\cong \mathbb{C}^{q^n})$ be the $q^n$-dimensional Hilbert space over $\mathbb{C}$. Similar to classical linear codes, 
a $q$-ary quantum code $\mathcal{Q}$ of length $n$ is a subspace of $\mathbb{C}^{q^n}$. 
Let $\{|\mathbf{a}\rangle=|a_1\rangle\otimes \dots \otimes |a_n\rangle:\ (a_1,a_2, \dots, a_n)\in \F_q^n \}$ be a basis of $\mathbb{C}^{q^n}$. 
Similarly, an $n$-qubit is a joint state of $n$ qubits in $\mathbb{C}^{q^n}$ and can be written as 
$$|\mathbf{v}\rangle =\sum_{\mathbf{a}\in \F_q}v_\mathbf{a}|\mathbf{a}\rangle,$$
where $v_{\mathbf{a}}\in \mathbb{C}$.
The Hermitian inner product of any two $n$-qubits $|\mathbf{u}\rangle =\sum_{\mathbf{a}\in \F_q}u_\mathbf{a}|\mathbf{a}\rangle$
and $|\mathbf{v}\rangle =\sum_{\mathbf{a}\in \F_q}v_\mathbf{a}|\mathbf{a}\rangle$ is defined by 
$$\langle \mathbf{u}|\mathbf{v} \rangle= \sum_{\mathbf{a}\in\F_q^n}u_{\mathbf{a}} {\bar{v_\mathbf{a}}}\in \C,$$
where $\bar{v_{\mathbf{a}}}$ is the complex conjugate of $v_\mathbf{a}$.
$|\mathbf{u}\rangle$ and $|\mathbf{v}\rangle$ are said to be orthogonal if $\langle \mathbf{u}|\mathbf{v} \rangle=0$.

Let $\zeta_p$ be a complex primitive $p$-th root of unity.
The actions (rules) of $X(\mathbf{a})$ and $Z(\mathbf{b})$ on $|\mathbf{v}\rangle\in \mathbb{C}^{q^n} (\mathbf{v}\in \F_q^n)$ are depicted as   
$$X(\mathbf{a})| \mathbf{v}\rangle=|\mathbf{v}+\mathbf{a}\rangle\quad {\rm{and}} \quad Z(\mathbf{b})|\mathbf{v}\rangle =\zeta_p^{tr(\langle\mathbf{v},\mathbf{b}\rangle_E)}|\mathbf{v},$$
respectively, where $tr(\cdot)$ is the trace function from $\F_q$ to $\F_p$.
In a quantum system, the quantum errors are some unitary operators.
Denote the error group by $G_n$, then  
$$G_n=\{\zeta_p^tX(\mathbf{a})Z(\mathbf{b}):\ \mathbf{a},\mathbf{b}\in \F_q^n, t\in \F_p\}. $$ 
For any error $E=\zeta_p^tX(\mathbf{a})Z(\mathbf{b})\in G_n$,
we define the quantum weight of $E$ as 
$$wt_Q(E)= \sharp \{i:(a_i,b_i)\neq (0,0)\},$$ 
where $\sharp$ denotes the number of elements in the set.
A quantum code $\mathcal{Q}$ with dimension $K\geq 2$ is said to detect $d-1$ quantum errors ($d\geq 1$),
if for any pair $|\mathbf{u}\rangle$ and $|\mathbf{v}\rangle$ in $\mathcal{Q} $ with $\langle \mathbf{u}|\mathbf{v}\rangle=0$ 
and any $E\in G_n$ with $wt_Q(E)\leq d-1$, we have $\langle \mathbf{u}|E|\mathbf{v}\rangle=0$. For a $q$-ary quantum code of length $n$,
dimension $K$ and minimum distance $d$, we usually denote it by $((n,K,d))_q$ or $[[n,k,d]]_q$, where $k=\log_q K$. 

Let $S$ be an abelian subgroup of $G_n$. Then the quantum stabilizer codes $C(S)$ can be defined by  
$$C(S)=\{|\phi\rangle:\ E|\phi \rangle= |\phi\rangle, \forall E\in S \},$$ 
which are analogues of classical additive codes. As we mentioned before, from classical linear codes satisfying certain orthogonality, by 
the method (namely, CSS construction) introduced by Calderbank et al. \cite[]{RefJ17} and Steane \cite[]{RefJ18}, one can obtain quantum stabilizer codes. 
However, this method fails when $S$ is a non-abelian. By extending $S$ to be a new abelian subgroup in a larger error group and assuming that both sender 
and receiver shared the pre-existing entangled bits, Burn et al. \cite[]{RefJ19} introduced EAQECCs. In this case, EAQECCs can be derived from any classical 
linear codes.

\subsection{New EAQECCs and MDS EAQECCs of length $n>q+1$}
Based on the known method of constructing EAQECCs, we use MDS codes with Hermitian hulls of flexible dimensions obtained in Section \ref{sec-construction} to 
obtain new EAQECCs and MDS EAQECCs. Like classical linear codes, there exists a trade-off between the parameters $n$, $k$, $d$ and $c$ of an EAQECC, called quantum Singleton bound.
\begin{lemma}\label{lem2.5}(Quantum Singleton bound \cite{RefJ32})
  Let $\mathcal{Q}$ be an $[[n,k,d;c]]_q$ EAQECC. If $2d\leq n+2$, then
\[\begin{split}
  k\leq n+c-2(d-1).
\end{split} \]
\end{lemma}
\begin{remark}\label{rem4.EAQMDS}
  \begin{enumerate}
    \item [{\rm (1)}] An EAQECC for which equality holds in this bound, i.e., $2d\leq n+2$ and $k=n+c-2(d-1)$, is called an MDS EAQECC.
    \item [{\rm (2)}] It is well known that if a classical linear code $\C$ is MDS and $d\leq \frac{n+2}{2}$, then the EAQECC constructed by it is an MDS EAQECC.
  \end{enumerate}
\end{remark}

For a matrix $M=(m_{ij})$ over $\F_{q^2}$, we denote the conjugate transpose of $M$ by $M^\dag=(m^q_{ji})$. In practice, the explict method of constructing  
EAQECCs from a linear code with certain dimensional Hermitian hull was established by Galindo et al. in \cite{RefJ33}. We rephrase the important result 
in the following. 

\begin{lemma}\label{lem.EAQECCs Construction}(\cite{RefJ33})
Let $H$ be a parity check matrix of a $q^2$-ary $[n,k,d]$ linear code. Then there exists an $[[n,2k-n+c,d;c]]_q$ EAQECC $\mathcal{Q}$, 
where $c={\rm{rank}}(HH^\dag)$ is the required number of maximally entangled states.
\end{lemma}

Generally speaking, the determination of the number of $c$ is difficult. Guenda et al. \cite{RefJ4} proposed a relationship between $c$ and $\dim(\Hull_H(\C))$ 
as follows, which simplifies the problem of calculating the number of $c$.

\begin{lemma}\label{lem.c calculation}(\cite{RefJ4})
 Let $\C$ be a $q^2$-ary $[n,k,d]$ linear code and $H$ be a parity check matrix of $\C$. Then 
\[\begin{split}
  rank(HH^\dag) = & n-k-\dim(\Hull_H(\C))\\
                = & n-k-\dim(\Hull_H(\C^{\bot_H}))\\
\end{split} \]
\end{lemma}

Since the Hermitian dual code of an $[n,k,n-k+1]_{q^2}$ MDS code is an $[n,n-k,k+1]_{q^2}$ MDS code and $(\C^{\bot_H})^{\bot_H}=\C$, by the result $(2)$ 
of Remark \ref{rem4.EAQMDS} and Lemma \ref{lem.EAQECCs Construction}, we can obtain the following result immediately. 

\begin{lemma}\label{lem.EAQMDS Construction}
  Let $H$ be a parity check matrix of a $q^2$-ary $[n,k,d]$ linear code and $l=\dim(\Hull_H(\C))$. If $k\leq \lfloor \frac{n}{2} \rfloor$, then there exists 
  an $[[n,k-l,n-k+1;n-k-l]]_q$ EAQECC $\mathcal{Q}$ and an $[[n,n-k-l,k+1;k-l]]_q$ MDS EAQECC $\mathcal{Q'}$.
\end{lemma}

Now, according to Lemma \ref{lem.EAQMDS Construction}, we can present our new constructions of $q$-ary EAQECCs and MDS EAQECCs.

\begin{theorem}\label{th.EAQECC1}
  Let $q=p^m\geq 3$ be a prime power. The following statements hold.
  \begin{enumerate}
    \item [{\rm (1)}] For $0\leq k\leq q-1$, if $q^2-k\leq n\leq q^2$, then there exists an $[[n,k-l,n-k+1;n-k-l]]_q$ EAQECC $\mathcal{Q}$ 
    and an $[[n,n-k-l,k+1;k-l]]_q$ MDS EAQECC $\mathcal{Q'}$, where $0\leq l\leq n+k-q^2$.
    \item [{\rm (2)}] For $0\leq k\leq q-1$, if $q^2-k+1\leq n\leq q^2$, then there exists an $[[n+1,k-l,n-k+2;n+1-k-l]]_q$ EAQECC $\mathcal{Q}$ 
    and an $[[n+1,n+1-k-l,k+1;k-l]]_q$ MDS EAQECC $\mathcal{Q'}$, where $0\leq l\leq n+k-q^2-1$.
    \item [{\rm (3)}] If $q^2-q \leq n\leq q^2$, then there exists an $[[n+1,q-l,n-q+2;n+1-q-l]]_q$ EAQECC $\mathcal{Q}$ 
    and an $[[n+1,n+1-q-l,q+1;q-l]]_q$ MDS EAQECC $\mathcal{Q'}$, where $0\leq l\leq n+q-q^2$.
  \end{enumerate}
\end{theorem}

\begin{example}\label{exam.EAQECC1}
The results {\rm{(1),\ (2)\ and\ (3)}} of Theorem \ref{th.EAQECC1} can be used to obtain many EAQECCs and MDS EAQECCs. 
According to Remark \ref{rem1.conA.comparion}, the length of these new MDS EAQECCs are always greater than $q+1$ and do 
not completely cover each other. As an intuitive example, we list some new MDS EAQECCs in Table \ref{tab:EAQECCs1}. 
\end{example}
\begin{table}[!htb]
  \caption{Some EAQMDS codes construted by Theorem \ref{th.EAQECC1} over $\F_9$}
  \label{tab:EAQECCs1}       
  \begin{center}
    \begin{tabular}{cccc|cccc}
      \hline
       $k$ & $l$ & EAQMDS codes & Ref. & $k$ & $l$ & EAQMDS codes & Ref. \\
      \hline
      8 & 3 & $[[76,65,9;5]]_9$ & Theorem \ref{th.EAQECC1}(1) & 8 & 3 & $[[77,66,9;5]]_9$ & Theorem \ref{th.EAQECC1}(1)\\
      8 & 3 & $[[78,67,9;5]]_9$ & Theorem \ref{th.EAQECC1}(1) & 8 & 3 & $[[79,68,9;5]]_9$ & Theorem \ref{th.EAQECC1}(1)\\
      8 & 3 & $[[80,69,9;5]]_9$ & Theorem \ref{th.EAQECC1}(1) & 8 & 5 & $[[78,65,9;3]]_9$ & Theorem \ref{th.EAQECC1}(1)\\
      8 & 5 & $[[79,66,9;3]]_9$ & Theorem \ref{th.EAQECC1}(1) & 8 & 5 & $[[80,67,9;3]]_9$ & Theorem \ref{th.EAQECC1}(1)\\

      8 & 3 & $[[81,70,9;5]]_9$ & Theorem \ref{th.EAQECC1}(2) & 8 & 5 & $[[81,68,9;3]]_9$ & Theorem \ref{th.EAQECC1}(2)\\

      9 & 3 & $[[76,64,10;6]]_9$ & Theorem \ref{th.EAQECC1}(3) & 9 & 3 & $[[77,65,10;6]]_9$ & Theorem \ref{th.EAQECC1}(3)\\
      9 & 3 & $[[78,66,10;6]]_9$ & Theorem \ref{th.EAQECC1}(3) & 9 & 3 & $[[79,67,10;6]]_9$ & Theorem \ref{th.EAQECC1}(3)\\
      9 & 3 & $[[80,68,10;6]]_9$ & Theorem \ref{th.EAQECC1}(3) & 9 & 3 & $[[79,65,10;4]]_9$ & Theorem \ref{th.EAQECC1}(3)\\
      9 & 5 & $[[78,64,10;4]]_9$ & Theorem \ref{th.EAQECC1}(3) & 9 & 5 & $[[79,65,10;4]]_9$ & Theorem \ref{th.EAQECC1}(3)\\
      9 & 5 & $[[80,66,10;4]]_9$ & Theorem \ref{th.EAQECC1}(3) & 9 & 5 & $[[81,67,10;4]]_9$ & Theorem \ref{th.EAQECC1}(3)\\
      \hline
    \end{tabular}
  \end{center}
  \end{table}

\begin{theorem}\label{th.EAQECC2}
  Let $q=p^m\geq 3$ be a prime power and $n=m(q-1)$, where $2\leq m\leq q$. Then for any $1\leq k\leq m-1$, there exists an 
  $[[n+1,k-l,n-k+2;n+1-k-l]]_q$ EAQECC $\mathcal{Q}$ and an $[[n+1,n+1-k-l,k+1;k-l]]_q$ MDS EAQECC $\mathcal{Q'}$, where $0\leq l\leq k-1$.
\end{theorem}

\begin{example}\label{exam.EAQECC2}
  According to Remark \ref{rem1.comparing Hermitian self-orthogonal}, the EAQECCs and MDS EAQECCs of length $n> q+1$ constructed by Theorem \ref{th.EAQECC2} are new. 
  According to Remark \ref{rem.conB.length>q+1}, most of these EAQECCs and MDS EAQECCs can not be obtained by Theorem \ref{th.EAQECC1}.  
  We list some of them in Table \ref{tab:EAQECCs2}.
\end{example}
\begin{table}[!htb]
  \caption{Some EAQMDS codes construted by Theorem \ref{th.EAQECC2} for some $q$}
  \label{tab:EAQECCs2}       
  \begin{center}
    \begin{tabular}{ccccc|ccccc}
      \hline
       $q$ & $m$ & $k$ & $l$ & EAQMDS codes & $q$ & $m$ & $k$ & $l$ & EAQMDS codes \\
      \hline
      8 & 5 & 4 & 2 & $[[36,30,5;2]]_8$ & 8 & 6 & 4 & 2 & $[[43,37,5;2]]_8$\\
      8 & 7 & 6 & 4 & $[[50,40,7;2]]_8$ & 8 & 8 & 6 & 4 & $[[57,47,7;2]]_8$\\

      9 & 5 & 4 & 2 & $[[41,35,5;2]]_9$ & 9 & 7 & 4 & 2 & $[[57,51,5;2]]_9$\\
      9 & 9 & 8 & 3 & $[[73,62,9;5]]_9$ & 9 & 9 & 8 & 5 & $[[73,60,9;3]]_9$\\

      16 & 7 & 6 & 2 & $[[106,98,7;4]]_{16}$ & 16 & 8 & 6 & 2 & $[[121,113,7;4]]_{16}$\\
      16 & 13 & 6 & 2 & $[[196,188,7;4]]_{16}$ & 16 & 14 & 6 & 2 & $[[211,203,7;4]]_{16}$\\

      25 & 14 & 4 & 2 & $[[337,331,5;2]]_{25}$ & 25 & 18 & 4 & 2 & $[[433,427,5;2]]_{25}$\\
      25 & 15 & 4 & 2 & $[[361,355,5;2]]_{25}$ & 25 & 19 & 4 & 2 & $[[457,451,5;2]]_{25}$\\
      25 & 16 & 6 & 4 & $[[385,375,7;2]]_{25}$ & 25 & 20 & 6 & 4 & $[[481,471,7;2]]_{25}$\\
      25 & 17 & 6 & 4 & $[[409,399,7;2]]_{25}$ & 25 & 21 & 6 & 4 & $[[505,495,7;2]]_{25}$\\
      \hline
    \end{tabular}
  \end{center}
  \end{table}

\section{Summary and concluding remarks}\label{sec-conclusion}
The main contribution of this paper is to construct several new classes of MDS codes and totally determine their 
Euclidean hulls (See Theorems \ref{th.ConA_form of Euclidean} and \ref{th.ConC}) or 
Hermitian hulls (See Theorems \ref{th.ConA.1}, \ref{th.ConA.2}, \ref{th.ConA.3} and Corollary \ref{coro.1.flexible Hermitian hulls from 1.Hermitian almost self-orthogonal}). 
For Hermitian cases, four new classes of $q$-ary EAQECCs and four new classes of $q$-ary MDS EAQECCs of length $n>q+1$ are further obtained 
(See Theorems \ref{th.EAQECC1} and \ref{th.EAQECC2}).  
And for Euclidean cases, some new Euclidean self-orthogonal and one-dimensional Euclidean hull MDS codes are given as examples. 

In particular, for convenience, MDS codes with $(k-1)$-dimensional Hermitian hull are called Hermitian almost self-orthogonal MDS codes in this paper. 
Some new criterions for extended GRS codes being Hermitian almost self-orthogonal MDS codes and Hermitian self-orthogonal MDS codes are presented 
(See Lemma \ref{lem.Hermitian almost self-orthogonal EGRS} and Corollary \ref{coro.Hermitian self-orthogonal EGRS}). For future research, it would be 
interesting to construct more Hermitian (almost) self-orthogonal MDS codes and MDS EAQECCs of length $n>q+1$.

\section*{Acknowledgments}
This research was supported by the National Natural Science Foundation of China (Nos.U21A20428 and 12171134).
\section*{}


\begin{thebibliography}{}
  \bibitem{RefJ32} A. Allahmadi, A. AlKenani, R. Hijazi, N. Muthana, F. Özbudak, P. Solé, New constructions of entanglement-assisted quantum codes, Cryptogr. Commun. 14(1) (2022),15-37.
  \bibitem{RefJ1} E.F. Assmus, J.D. Key, Designs and Their Codes, Cambridge Univ, Press. (1993),103.


  \bibitem{RefJ19} T. Brun, I. Devetak, M.H. Hsieh, Correcting quantum errors with entanglement, Science 314(5798) (2006),436-439.


  \bibitem{RefJ17} A. Calderbank, P. Shor, Good quantum error-correcting codes exist, Phys. Rev. A, Gen. Phys. 54(2) (1996),1098-1105.
  \bibitem{RefJ28} B. Chen, H. Liu, New constructions of MDS codes with complementary duals, IEEE Trans. Inf. Theory 64(8) (2018),5776-5782.
  \bibitem{RefJ15} C. Carlet, C. Li, S. Mesnager, Linear codes with small hulls in semi-primitive case, Des. Codes Cryptogr. 87 (2019),3063-3075.
  \bibitem{RefJ34} H. Chen, New MDS Entanglement-Assisted Quantum Codes from MDS Hermitian Self-Orthogonal Codes, arXiv:2206.13995 [cs.IT], (2022). 


  \bibitem{Ref xin4} J. Fan, H. Chen, J. Xu, Constructions of $q$-ary entanglement-assisted quantum MDS codes with minimum distance greater than $q+1$, Quantum Inf. Comput. 16 (2016),423-434.
  \bibitem{RefJ29} W. Fang, F. Fu, Two new classes of quantum MDS codes, Finite Fields Appl. 53 (2018),85-98.
  \bibitem{RefJB1} W. Fang, F. Fu, Some New Constructions of Quantum MDS Codes, IEEE Trans. Inf. Theory 65(12) (2019),7840-7847. 
  \bibitem{RefJ7}  W. Fang, F. Fu, L. Li, S. Zhu, Euclidean and hermitian hulls of mds codes and their applications to eaqeccs, IEEE Trans. Inf. Theory 66(6) (2020),3527-3527.


  \bibitem{RefJ33} C. Galindo, F. Hernando, R. Matsumoto, D. Ruano, Entanglement-assisted quantum error-correcting codes over arbitrary finite fields, Quantum Inf. Process. 18(4) (2019),1-18.
  \bibitem{RefJ35} G. Guo, R. Li, Y. Liu, Application of Hermitian self-orthogonal GRS codes to some quantum MDS codes, Finite Fields Appl. 76 (2021),101901.
  \bibitem{RefJ5} K. Guenda, T.A. Gulliver, S. Jitman, S. Thipworawimon, Linear $l$-intersection pairs of codes and their applications, Des. Codes Cryptogr. 88(1) (2020),133-152.
  \bibitem{RefJ4} K. Guenda, S. Jitman, T. A. Gulliver, Constructions of good entanglement-assisted quantum error correcting codes, Des. Codes Cryptogr. 86 (2018),121-136.
	


  \bibitem{RefJB5} X. He, L. Xu, H. Chen, New q-ary quantum MDS codes with distances bigger than $q^2$, Quantum Inf Process. 15 (2016),2745-2758.

  \bibitem{RefJB4} L. Jin, H. Kan, J. Wen, Quantum MDS codes with relatively large minimum distance from Hermitian self-orthogonal codes, Des. Codes Cryptogr. 84 (2017),463-471. 
  \bibitem{RefJB3} L. Jin, C. Xing, A Construction of New Quantum MDS Codes, IEEE Trans. Inf. Theory 60(5) (2014),2921-2925.
  \bibitem{Ref xin2} R. Jin, D. Xie, J. Luo, New classes of entanglement-assisted quantum MDS codes, Quantum Inf. Process. 19(9) (2020),1-12.
  \bibitem{RefJ21} R. Jin, Y. Cao, J. Luo, Entanglement-assisted quantum MDS codes from generalized Reed-Solomon codes, Quantum Inf. Process. 20(2) (2021),1-17.


	\bibitem{RefJ12} C. Li, P. Zeng, Constructions of linear codes with one-dimensional hull, IEEE Trans. Inf. Theory 65(3) (2018),1668-1676.
  \bibitem{RefJ24} G. Luo, X. Cao, Two new families of entanglement-assisted quantum MDS codes from generalized Reed-Solomon codes, Quantum Inf. Process. 18(3) (2019),89.
	\bibitem{RefJ6} G. Luo, X. Cao, X. Chen, MDS codes with hulls of arbitrary dimensions and their quantum error correction, IEEE Trans. Inf. Theory 65(5) (2018),2944-2952.
  \bibitem{Ref xin3} L. Li, S. Zhu, L. Liu, Three new classes of entanglement-assisted quantum MDS codes from generalized Reed-Solomon codes, Quantum Inf. Process. 18(12) (2019),1-16.
	\bibitem{RefJ2} J. Leon, Computing automorphism groups of error-correcting codes, IEEE Trans. Inf. Theory 28(3) (1982),496-511.
	\bibitem{RefJ3} J. Leon, Permutation group algorithms based on partition, I: Theory and algorithms, J. Symb. Comput. 12 (1991),533-583.
  \bibitem{RefJ23} Y. Liu, R. Li, L. Lv, and Y. Ma, Application of constacyclic codes to entanglement-assisted quantum maximum distance separable codes, Quantum Inf. Process. 17(8) (2018),210.


	\bibitem{RefJ13} L. Qian, X. Cao, W. Lu et al, A new method for constructing linear codes with small hulls, Des. Codes Cryptogr. (2021).
  \bibitem{RefJ31} L. Qian, X. Cao, S. Mesnager, Linear codes with one-dimensional hull associated with Gaussian sum, Cryptogr. Commun. (2020).
	\bibitem{RefJ20} J. Qian, L. Zhang, On MDS linear complementary dual codes and entanglement-assisted quantum codes, Des. Codes Cryptogr. 86(7) (2018),1565-1572.
  \bibitem{Ref xin1} J. Qian, L. Zhang, Constructions of new entanglement-assisted quantum MDS and almost MDS codes, Quantum Inf. Process. 18(3) (2019),1-12.



	\bibitem{RefJ18} A. M. Steane, Error correcting codes in quantum theory, Phys. Rev. Lett. 77(5) (1996),793-797.
  \bibitem{RefJ30} L. Sok, A new construction of linear codes with one-dimensional hull, Des. Codes Cryptogr. (2022).
  \bibitem{RefJ3'} N. Sendrier, Finding the permutation between equivalent codes: the support splitting algorithm, IEEE Trans. Inf. Theory 46(4) (2000),1193-1203.


  \bibitem{RefJ27} G. Wang, C. Tang, Application of GRS codes to some entanglement-assisted quantum MDS codes, Quantum Inf. Process. 21(98) (2022).
	\bibitem{RefJ22} L. Wang, S. Zhu, New quantum MDS codes derived from constacyclic codes, Quantum Inf. Process. 14(3) (2015),881-889.	
	\bibitem{RefJ16} Y. Wang, R. Tao, Constructions of linear codes with small hulls from association schemes, Adv. Math. Commun. 16(2) (2022),349-364.




\end{thebibliography}
\end{document}